\documentclass{article}

\PassOptionsToPackage{longnamesfirst, authoryear}{natbib}

\usepackage{amsmath}
\usepackage{amsthm}
\usepackage{amssymb}
\usepackage{hyperref}
\usepackage{cleveref}
\usepackage[linesnumbered,lined,boxed,commentsnumbered]{algorithm2e}
\usepackage{mathtools}
\usepackage{comment}
\usepackage{subcaption}

\usepackage[preprint]{neurips_2026}

\usepackage[utf8]{inputenc} 
\usepackage[T1]{fontenc}    
\usepackage{hyperref}       
\usepackage{url}            
\usepackage{booktabs}       
\usepackage{amsfonts}       
\usepackage{nicefrac}       
\usepackage{microtype}      
\usepackage{xcolor}         

\usepackage[normalem]{ulem} 

\usepackage{tikz}
\usepackage{pgfplots}
\pgfplotsset{compat=1.18}

\theoremstyle{definition}
\newtheorem{definition}{Definition}

\theoremstyle{plain}
\newtheorem{theorem}{Theorem}
\newtheorem{proposition}[theorem]{Proposition}
\newtheorem{lemma}[theorem]{Lemma}
\newtheorem{corollary}[theorem]{Corollary}

\theoremstyle{remark}

\newcommand{\N}{\mathbb{N}}
\newcommand{\Z}{\mathbb{Z}}
\newcommand{\R}{\mathbb{R}}
\newcommand{\M}{\mathcal{M}}
\newcommand{\MD}{\mathcal{M}_{\text{DLap}}}
\newcommand{\E}[2][]{\mathbb{E}_{#1} \left[ #2 \right]}
\newcommand{\Var}[1]{\operatorname{Var} \left( #1 \right)}

\newcommand{\proba}[1]{\Pr \left[ #1 \right]}
\newcommand{\MSE}[2][]{\text{MSE}_{#1}(#2)}

\def\tx{\tilde{x}}

\def\vol{\textnormal{vol}}

\SetKwFunction{DLap}{DLap}

\renewcommand{\epsilon}{\varepsilon}

\author{%
  Quentin Hillebrand \\
  BARC \\
  University of Copenhagen\\
    \And
  Jacob Imola \\
  University of Waterloo \\
    \AND
  Rasmus Pagh \\
  BARC \\
  University of Copenhagen\\
    \And
  Sia Sejer \\
  BARC \\
  University of Copenhagen\\
}

\title{Privacy by Postprocessing the\\ Discrete Laplace Mechanism}

\begin{document}

\maketitle

\begin{abstract}
    We show that an ``old dog,'' the classical discrete Laplace (aka.~geometric) mechanism, can ``perform new tricks'':
    \begin{enumerate}
        \item It can be post-processed to yield a simple, unbiased estimator of any subexponential function~$f$ of the original data, giving a simple, discrete, multivariate version of the recent unbiasing result for the Laplace mechanism by Calmon et al.~(FORC~'25).
        \item It can be post-processed to output the same distribution as the Laplace mechanism or the Staircase mechanism with identical privacy parameters.
    \end{enumerate}
    Thus, the discrete Laplace mechanism is a versatile mechanism that should be preferred over the Laplace and Staircase mechanisms whenever the data is discrete (or can be made discrete while controlling $\ell_1$-sensitivity).
    
    We show bounds on the variance of our estimator, compared to the mean square error of the biased estimator that simply evaluates the $f$ on the output of the mechanism.
    Though our unbiased estimator has exponential running time for worst-case functions, we show that it can often be computed in linear or polynomial time for some common functions exhibiting structure.
    We showcase the properties of our methods empirically with several use cases including profile and entropy estimation, as well as distributed/federated data analysis applications in which unbiasedness is key to accuracy.
\end{abstract}

\section{Introduction} 
Differential privacy has a powerful \emph{post-processing} property: 
For any differentially private mechanism~$\M$ and any function~$f$, we can release $f(\M(x))$ while preserving the privacy guarantees of $\M$. 
As a consequence, we can release all values $f_i(\M(x))$ for any family of functions $(f_i)_{i\in I}$ with no loss of privacy. 
In this paper, we consider additive noise mechanisms where $\E{\M(x)}=x$.
For non-linear functions~$f$, the expected function value $\E{f(\M(x))}$, assuming it exists, is in general different from $f(x)$. 
Such bias can be a problem in some applications, which led \cite{hillebrand2023unbiased} and later \cite{calmon2025debiasing} to study methods for debiasing via post-processing of the Laplace mechanism.
First, \cite{hillebrand2023unbiased} showed how to construct unbiased estimators for multivariate polynomials.
Later, \cite{calmon2025debiasing} gave estimators in the univariate setting for a wide class of twice differentiable functions~$f$. 

\paragraph{Debiasing in Post-Processing.}
In this paper, we propose a new post-processing method for the \emph{discrete} Laplace mechanism $\MD: \Z^n \rightarrow \Z^n$, first proposed and studied by \citet{ghosh2009universally}.
For convenience we will use the notation $\tilde{x} = \MD(x)$, that is, $\tilde{x}_i = x_i + \eta_i$ where $\eta_i$ is sampled i.i.d.~from the discrete Laplace distribution with parameter $p = e^{-\varepsilon}$. 
This mechanism is a canonical way of releasing histograms, and more generally integer-valued data with sensitivity~1, with an $\varepsilon$-differential privacy guarantee.
Given $\tilde{x}$, our post-processing method returns an unbiased estimator of $f(x)$ for \emph{any} univariate or multivariate function~$f$ where $\E{f(\tilde x)}$ exists and is finite.
A sufficient condition for this is that~$f$ is subexponential (or exponential with low enough rate).
Our method thus supports a class of functions~$f$ much wider than that of \cite{calmon2025debiasing, hillebrand2023unbiased}, both of which can only handle functions that are twice differentiable. 
Previously, unbiased postprocessing of the discrete Laplace distribution had only been described in the special case of univariate threshold functions, by \citet{ghazi2022anonymized}.

We show that our post-processing method is unique in the sense that it is the only deterministic method that constructs unbiased estimates in the discrete setting.
From this, it follows that it has the smallest variance of all unbiased estimators.
Though our estimator takes exponential time to compute for worst-case functions, we show that for functions which exhibit structure (e.g.~decomposable functions such as entropy and order statistics), it is often possible to design linear or polynomial time algorithms for computing the unbiased estimation.

\paragraph{Transformations to other Additive Noise Mechanisms.}
We show how to post-process $\tilde{x}$ so that the resulting distribution is identical to that of the classical (real-valued) Laplace mechanism~\citep{dwork2006calibrating} or the Staircase mechanism~\citep{geng2014optimal,geng2015optimal} on input $x$ with the same privacy parameter $\varepsilon$.
This means the Discrete Laplace Mechanism is a versatile mechanism that should be preferred over the Laplace and Staircase Mechanisms whenever the data is discrete or can be made discrete while still controlling the $\ell_1$-sensitivity.
Everything that can be done with these mechanisms can also be based on the Discrete Laplace Mechanism, and the latter may lead to better utility, especially for medium to large values of $\varepsilon$.

\paragraph{Applications to Federated/ Distributed Settings.}
Unbiased estimators are particularly useful in settings where the statistic of interest can be obtained by summing or averaging many individual contributions.
For example, suppose that each of $k$ parties (e.g.~hospitals) holds a distribution given by a histogram, that each histogram has been released with the Discrete Laplace mechanism, and we wish to compute the entropy of the joint distribution over all parties.
Assuming independence, the joint entropy is the sum of the entropy of each distribution, a non-linear function of the histogram.
Summing an unbiased estimator for the entropy of each distribution we get an estimator for the total entropy whose error scales with $O(\sqrt{k})$.
In contrast, summing biased estimators will result in error proportional to $k$, which will be worse for sufficiently large $k$.
As another example, suppose we wish to gather statistics \emph{across} histograms, for example the minimum, maximum, mean and variance of each column in the histograms.
We can avoid bias in these estimates by using our estimator on the ``transposed'' histograms formed by combining information on each column across histograms.

\subsection{Overview of Our Contributions.}

In \Cref{sec:DebiasDLap}, we first show (\Cref{thm:unbiase-univariate}) that in the univariate case ($n=1$) a linear combination of $f(\tilde{x}-1)$, $f(\tilde{x})$, and $f(\tilde{x}+1)$ is an unbiased estimator of $f(x)$ for any function $f$ and input~$x$ where $\E{f(\tilde{x})}$ exists.
This estimator can be interpreted as the discrete second derivative of $f$ evaluated in $\tilde{x}$, analogous to the approach by \cite{calmon2025debiasing} for univariate functions.
We then show (\Cref{thm:unbiase-multivariate}) that our method extends naturally to the multivariate case $x\in\mathbb Z^n$ where the estimator is a linear combination of the function values $f(\tilde x+\xi)$ for $\xi \in \{-1,0,+1\}^n$.

\Cref{sec:properties} investigates properties of our unbiased estimator.
We first show (\Cref{thm:uniqueness}) that it is unique among deterministic estimators and has minimum variance among all unbiased estimators.
In particular, for a black-box function $f$, any deterministic estimator requires $3^n$ function evaluations.
Then we consider the mean square error (MSE) and show (\Cref{thm:l2-bound-multivariate}) that it is at most a factor $O(1+1/\varepsilon^{4})^n$ worse than the MSE of the naive estimator~$f(\tilde{x})$.
This bound is tight: Proposition~\ref{prop:power-func} shows that there are functions for which our estimator exhibits this rate of exponential variance growth, but for other functions we instead see that the unbiased estimator is exponentially \emph{better}.

\Cref{sec:tranformation} shows that it is possible to post-process the Discrete Laplace Mechanism to be identical to either the classical Laplace Mechanism (\Cref{thm:laplace}) or to the Staircase mechanism (\Cref{thm:staircase}) while preserving the same privacy level $\varepsilon$.
The insight is that it is possible to choose random variables $X_\varepsilon$ and $Y_{\gamma, \varepsilon}$ (independently of $x$), such that $\tilde{x} + X_\varepsilon$ and $\tilde{x} + Y_{\gamma, \varepsilon}$ have the same distribution as the classical Laplace mechanism and the Staircase Mechanism with privacy parameter~$\varepsilon$, respectively.  

\Cref{sec:faster} studies functions and classes of functions for which our estimator can be computed much faster than by direct summation of its $3^n$ terms.
The key idea is to group terms with the same contribution, and compute how many terms are in each group.
This allows us to compute order statistics (such as the minimum, median, and maximum), inference for decision trees, as well as entropy and KL divergence in polynomial time.
 
In \Cref{sec:Empirical} we present four case studies showcasing settings in which our estimator has an advantage over alternatives.
We consider polynomials, for a use case in graph counting, for which we emphasize the advantage of the discrete Laplace mechanism over the classical Laplace mechanism for high values of $\varepsilon$.
For empirical entropy, where we see a significantly improved mean square error for low- and medium values of $\varepsilon$ for sparse histograms.
For estimating partition functions, a setting in which the naive estimator has considerable bias, we also see improved behavior near the convergence radius.
Finally, we give the first practical demonstration of profile estimation by post-processing.

\paragraph{Societal Impact of Unbiasedness.}
So far, we have considered unbiasedness from the technical perspective of statistical bias. 
From a societal perspective, one can ask how unbiased estimation relates to algorithmic bias and fairness. 
There are many distinct and incompatible notions of algorithmic fairness (e.g. \cite{Kleinberg2017Fairness}). 
A potential benefit of pointwise unbiased estimators is that they avoid directional statistical bias so that no group of points is systematically underestimated or overestimated in expectation.
However, an unbiased estimator might have a high variance and over- or underestimate the function value most of the time, reducing its fairness in practice.
For example, the estimator may overestimate the target value with high probability while remaining unbiased, because it occasionally produces underestimates that compensate in expectation. 
Another concern is that the estimator's variance may differ substantially across groups of input points.
This raises the question of whether one can design unbiased estimators with uniform or controlled variance for all points, and what trade-offs such constraints would impose. 
To sum up, one should be careful not to conflate unbiasedness and fairness.

\subsection{Related Work}

\paragraph{Debiasing Functions of the Laplace Mechanism.}
\citet{hillebrand2023unbiased} showed how to debias functions of the classical Laplace mechanism (an additive noise mechanism that yields privacy guarantees also for non-integer data) when $f$ is a multivariate polynomial.
\citet{calmon2025debiasing} considered a much more general class of functions in the univariate case ($n=1$). 
They showed that there is an essentially unique way of computing such unbiased estimators: for a suitable constant $c_\varepsilon$ and any twice-differentiable, so-called tempered function $f$, the estimator
\[\hat f(\mathcal{M}(x))=f(\mathcal{M}(x))-c_\varepsilon f''(\mathcal{M}(x))\]
\noindent
satisfies $\E{\hat f(\mathcal{M}(x))}=f(x)$.  
The result is derived via a deconvolution perspective in the Fourier domain and yields closed-form unbiased estimators for broad classes of nonlinear functions.
It is not immediately clear from their results how to extend the result to the multivariate setting. 
For our estimator these strong assumptions on $f$ are not necessary, so it is applicable to a wider class of functions for which the naive estimator has a well-defined, finite expectation.

\paragraph{Mechanisms based on Postprocessing the Discrete Laplace Mechanism.}
\citet{ghazi2022anonymized} considered a related distributed setting in which a cryptographic primitive, a \emph{shuffler}~\cite{cheu2019distributed, erlingsson2019amplification}, is used to collect and noise histogram data.
Similar to our results they consider Discrete Laplace noise, and in the univariate case they show how to debias a threshold indicator function applied to the output of the Discrete Laplace mechanism.
This result is used as a stepping stone to estimate any \emph{symmetric} function of the histogram.
For comparison we provide a method that can be used without shuffling and supports a wider class of functions that are not necessarily symmetric.

\citet{wu2024profile} studied \emph{profile estimation} which boils down to estimating for various $k$ a sum of indicator function values $f(x)$ where $f(x) = 1$ if $x=k$ and $f(x)=0$ otherwise.
They present a way of ``inverting'' the noise added by the Discrete Laplace mechanism and show that it is asymptotically near-optimal.
Our estimator gives a much simpler solution to this problem.

\paragraph{Privately Estimating Black-Box Statistics.}
Related to supporting privacy on a broad class of functions, \cite{linder2025untrustedBlackboxFunction} designed a so-called down-local algorithm for estimating a function $f$ with exponentially many evaluations of $f$ on noisy data.
Despite the superficial similarity, their setting is rather different from ours.
Their method provides accurate estimations for \textit{well-behaved} (finite-domain or $c$-Lipschitz) functions on the dataset and its subset. 
For monotone functions, the bias is always downwards, and for more general functions with finite range, they provided a guaranteed interval. \cite{steinke2025blackBoxStatistics} presents a method for interpolating between computational efficiency and statistical accuracy when estimating a black box function $f$. 
Formally, they interpolate between high statistical accuracy and an exponential number of evaluations, and low statistical accuracy, where the number of evaluations depends only on the privacy parameters and the domain size.

\subsection{Definitions}
\begin{definition}[Differential Privacy (DP) \citep{dwork2006calibrating}]
For $\varepsilon > 0$ an algorithm $\mathcal{M}$ with domain~$\mathcal{D}$ and range $\mathcal{S}$ satisfies $\varepsilon$-differential privacy ($\varepsilon$-DP) if, for every $S \subset \mathcal{S}$ and any $D^{(1)}, D^{(2)} \in \mathcal{D}$ that differ by only one element, we have
\[
    \Pr \left[ \mathcal{M}\left(D^{(1)}\right) \in S \right]
    \leq {e^{\varepsilon} \cdot \Pr \left[ \mathcal{M}\left(D^{(2)}\right) \in S \right]} \enspace .
\]
\end{definition}
In this paper we consider datasets that have been transformed to integer-valued vectors in $\Z^n$. 
A canonical example of such a dataset is a histogram of counts of elements in a discrete domain of size~$n$.
We consider two datasets $x, x'$ to be neighboring if $\|x - x'\|_1 \leq 1$, i.e., at most one count can be changed by $1$.
For discrete-valued data, the canonical DP mechanism is based on adding noise drawn from the \emph{discrete Laplace distribution}:

\begin{definition}[Discrete Laplace distribution]
    \label{def:dlap}
    For $p \in (0,1)$, $\DLap(p)$ denotes the discrete Laplace distribution whose probability mass at $i \in \mathbb{Z}$ is $\frac{1-p}{1+p} \cdot p^{|i|}$.
\end{definition}

We consider the setting where $x$ has been published as a noisy histogram $\tilde{x} = x + \DLap(p)^n$ by adding independent noise drawn from $\DLap(p)$ to each coordinate. 
This is an example of the Discrete Laplace Mechanism, and it is known that if $p = e^{-\varepsilon}$, it satisfies $\varepsilon$-DP~\citep{ghosh2009universally}. 
We are interested in estimating a function $f(x)$ of the sensitive dataset, where $f:\Z^n \rightarrow \R$, while only having access to $\tilde{x}$. 
To do so, we will use an estimator $g(\tilde{x})$.
We measure the accuracy of an estimator using the mean square error.

\begin{definition}[Mean Square Error]
    The mean square error (MSE) of an estimator $g$ on input $x$ is 
    \(
        \MSE[g]{x} = \E{(g(\tilde{x}) - f(x))^2}
    \).
    If the expected value does not exist the MSE is \emph{undefined}.
\end{definition}

Note that this definition includes the dataset $x$ as a parameter, rather than considering a maximum over all $x$. 
This pointwise definition will be useful when computing our error bounds later.
    
\section{Debiasing the Discrete Laplace Mechanism} \label{sec:DebiasDLap}

We now study how to design an unbiased estimator $g$ that, given access to the dataset $x$ published with Discrete Laplace Noise, has expected value $f(x)$ for a given function~$f$.

\paragraph{Univariate Case.}
We first consider a single dimension $n = 1$. The basic idea is to ``invert'' the equation $\E{g(\tilde{x})} = f(x)$, which is generally possible because, when expanded out, expectation is a linear operator. 
In general, such an estimator would not have a closed form, but the Discrete Laplace Distribution has tails with a simple exponential shape, which enables cancellation. This makes it possible to debias $f(\tilde{x})$ by considering only $f(\tilde{x} + 1)$, $f(\tilde{x})$, and $f(\tilde{x}-1)$. 
We show: 
\begin{theorem}\label{thm:unbiase-univariate}
    For $f: \Z \to \R$, $p \in (0,1)$, $x \in \Z$, and $y \in \Z$ define: 
    \begin{equation}
        g(y) = f(y) - \tfrac{p}{(1 - p)^2} \left( f(y+1) - 2 f(y) + f(y-1) \right).
        \label{eq:unbiased}
    \end{equation}
    Given $\tilde{x} = x + \DLap(p)$ if $\E{|f(\tilde x)|}$ exists and is finite then $g(\tilde{x})$ is an unbiased estimator of $f(x)$, i.e., $\E{g(\tilde{x})} = f(x)$.
\end{theorem}
The proof of \Cref{thm:unbiase-univariate} is a straightforward computation of the expected value of $g(\tilde{x})$ and is omitted to Appendix~\ref{app:debias}.
The condition that $\E{f(\tilde x)}$ exists and is finite is satisfied for every function whose growth rate is subexponential or, more generally, dominated by a certain exponential function depending on $p$, see Lemma~\ref{lem:poly} for details.

{\bf Comment.}
The value $f(y+1) - 2f(y) + f(y-1)$ is equal to  $(f(y+1) - f(y) )- (f(y) - f(y-1))$ and can be interpreted as a discrete second derivative. 
This makes our estimator analogous to the unbiased estimator for real-valued functions $f(\tilde{x}) - c_\varepsilon f''(\tilde{x})$ derived in~\cite{calmon2025debiasing}.

The special case where $f$ is a threshold function, i.e., where for some value $k$, $f(x)=1$ if $x\geq k$ and $f(x)=0$ for $x<k$, was studied previously in \citet[Lemma 10]{ghazi2022anonymized}.
The unbiased estimator given in their paper coincides with ours in this case.
In the univariate case, our estimator in \Cref{eq:unbiased} can in fact be recovered via indicator functions $f_k(x)=\mathbf{1}[x=k]$, combining the resulting linear estimators. Our approach operates directly at the level of the function $f$
and extends naturally to arbitrary multivariate functions via tensorization as we discuss below.

\paragraph{Multivariate Case.}
We next discuss how to extend Theorem~\ref{thm:unbiase-univariate} to multivariate functions.
\begin{theorem}
    \label{thm:unbiase-multivariate}
    For $f: \Z^n \to \R$, $p \in (0,1)$, $y \in \Z^n$, define: 
    \begin{equation}
        g(y) = \sum\limits_{\xi \in \{-1, 0, 1\}^n} \left(f(y + \xi) \prod\limits_{j=1}^n \alpha_{\xi_j}\right) \text{ where }
        \alpha_{\xi_i} =
        \begin{cases}
            1 + \frac{2 p}{(1-p)^2}, & \text{for } \xi_i=0 \\
            \frac{-p}{(1-p)^2},      & \text{for } \xi_i \in \{-1,+1\}
        \end{cases}
        \label{eq:unbiased-multivariate}
    \end{equation}
    For all $i \in [n]$ let $\tilde{x}_i = x_i + \DLap(p)$.
    If $\E{|f(\tilde x)|}$ exists and is finite, $g(\tilde{x}_1, \ldots, \tilde{x}_n)$ is an unbiased estimator of $f$, i.e., $\E{g(\tilde{x}_1, \ldots, \tilde{x}_n)} = f(x_1, \ldots, x_n)$
\end{theorem}
The proof of \Cref{thm:unbiase-multivariate}, which can be found in Appendix~\ref{app:debias}, essentially applies the reasoning of \Cref{thm:unbiase-univariate} to one variable at a time. 
A direct implementation of the estimator (\ref{eq:unbiased-multivariate}) sums $3^n$ terms, making it impractical for large $n$.
The sum can be estimated by sampling, but that results in a larger variance.
In Section~\ref{sec:faster} we show how to compute unbiased estimators for specific multivariate functions with much better computational complexity.

\section{Properties of the Unbiased Estimator}\label{sec:properties}

In this section we show some basic properties of our estimator.
First, we show that it is the unique deterministic, unbiased estimator (\Cref{sec:uniqueness}).
Second, we study the mean square error (MSE) of our estimator.
We show an upper bound in terms of $\MSE[f]{x}$, the MSE of the naive estimator (\Cref{sec:bounding-l2-error}), and give examples of functions for which the upper bound is essentially tight, and well as examples of functions where our estimator is much better than the naive estimator (\Cref{sec:gaps}).

\subsection{Uniqueness of the Unbiased Estimator}\label{sec:uniqueness}

We show that the estimator (\ref{eq:unbiased-multivariate}) is the only such unbiased, deterministic estimator for discrete Laplace noise. 
Moreover, if there exists another (potentially randomized) estimator, then its variance cannot be smaller than the one of ours.
To obtain those results, we use the Rao-Blackwell Theorem:

\begin{theorem}[Rao-Blackwell Theorem (\citet{lehmann1998theory}, Theorem 7.8)]
    \label{thm:rao}
    Let $Y$ be an observed random variable and let $U$ be auxiliary randomness.
    Let $\delta(Y,U)$ be an estimator with finite variance, and define
    \(
        \bar{\delta}(Y)
        \coloneqq
        \E{\delta(Y,U)\mid Y}.
    \)
    Then
    \[
        \E{\bar{\delta}(Y)} = \E{\delta(Y,U)}
        \quad\text{and}\quad
        \Var{\bar{\delta}(Y)} \leq \Var{\delta(Y,U)}.
    \]
\end{theorem}

We are now ready to state the main claim of the section:

\begin{theorem}\label{thm:uniqueness}
    Let $n \in \N^*$, $f: \Z^n \to \R$, $p \in (0,1)$, and $g$ the unbiased estimator of \Cref{eq:unbiased-multivariate}. 
    Assuming all expectations are well defined, $g$ is the only deterministic unbiased estimator of $f(x)$ from $\tx = x + \DLap(p)$.
    Moreover, if $g$ has a finite variance, then it has minimum variance among all unbiased estimators, including randomized estimators.
\end{theorem}

To prove uniqueness, observe that any two unbiased estimators $g_1, g_2$ for $f$ produce an unbiased estimator $g_1-g_2$ for $0$. 
We show that $g = 0$ is the only deterministic, unbiased estimator for $f = 0$ using cancellation properties. Optimality of $g$ follows directly from~\Cref{thm:rao}. 
Details can be found in Appendix~\ref{app:uniquenessproof}.

\subsection{Bounding MSE of the Unbiased Estimator} \label{sec:bounding-l2-error}

For a general function $f$, we bound $\MSE[g]{x}$ as follows:
\begin{theorem}
    \label{thm:l2-bound-multivariate}
    For $f: \Z^n \to \R$, $p \in (0,1)$, $(x_1, \ldots, x_n) \in \Z^n$, and for all $i \in [n]$, $\tilde{x}_i = x_i + \DLap(p)$, and $g$ as in \Cref{eq:unbiased-multivariate}. Assuming $f(\tx)$ has well defined expectation and variance,
    \[
        \MSE[g]{x_1, \ldots, x_n} \leq \left( 3 \tfrac{\left(1 + p^2\right)^2 + 2p}{(1-p)^4} \right)^n \MSE[f]{x_1, \ldots, x_n}.
    \]
\end{theorem}
To prove the theorem above, we start by showing a lemma that bounds $f((\tilde{K}+\xi)-f(K))^2$ in terms of the MSE of $f$, and then applying the lemma.
The details of the proof are included in Appendix~\ref{app:proof of l2-bound-multivariate}.
This result affirms that $\MSE[g]{x}$ can be bounded $\MSE[f]{x}$, though the factor $O(1+1/\varepsilon^4)^n$ is exponential in terms of $n$, and can also be high with low $\varepsilon$. In the next section, we derive an example where this is asymptotically tight, but also show it is possible for the \emph{other} direction to hold---i.e. for $\MSE[g]{x}$ to be exponentially \emph{lower} than $\MSE[f]{x}$.

\subsection{Exponential MSE Gaps Between Estimators}\label{sec:gaps}

In this section, we first study a family of functions for which the debiasing operator has an especially simple closed form. Let $f_{\gamma}(y) = \gamma^{S(y)}$, with $S(y) = \sum_i y_i$. These functions are useful test cases because their unbiased estimator is obtained by multiplying the naive estimator by a scalar factor $A(\gamma)^n$. This makes the dependence on the dimension n explicit and shows that the exponential factor in the general multivariate bound is unavoidable.

\begin{proposition}
    \label{prop:power-func}
    For $p \in (0,1)$, $\gamma \in (-1/p,-p) \cup (p, 1/p)$, $x \in \Z^n$, and $\tx = x + \DLap(p)$, $f_{\gamma}$ has a well defined expected value $\E{|f_{\gamma}(\tx)|}$ and the unbiased estimator defined in \Cref{eq:unbiased-multivariate} is $g_{\gamma}(y) = A(\gamma)^n f_{\gamma}(y)$ with 
    \[
        A(\gamma) = 1 - \tfrac{p}{(1-p)^2} \tfrac{(1-\gamma)^2}{\gamma}.
    \]
    If additionally $\sqrt{p} < |\gamma| < 1/\sqrt{p}$, then both variances exist and
    \(
        \Var{g_{\gamma}(\tx)} = A(\gamma)^{2n} \Var{f_{\gamma}(\tx)}
    \).
\end{proposition}

This result follows by simple computation. The full proof can be found in Appendix~\ref{app:proof-power-func}.

Note that for $\gamma < 0$ we have $A(\gamma) > 1$, and for $\gamma > 0$ we have $0 < A(\gamma) < 1$.
This means that unbiasing can both increase or reduce the variance of the estimator depending on the function:

\begin{corollary}
    For every $p \in (0,1)$, there are functions for which debiasing increases the MSE exponentially in $n$, and functions for which it decreases the MSE exponentially in $n$.
\end{corollary}

\begin{proof}
    First we notice that the MSE is equal to the variance for the unbiased estimator, but larger for the naive estimator.
    Proposition~\ref{prop:power-func} with $\gamma$ such that $\sqrt{p} < \gamma < 1/\sqrt{p}$, $\gamma \neq 1$ gives an exponential decrease in variance for the unbiased estimator, and thus at least an exponential decrease in the MSE.
    
    For the exponential increase, we set $\gamma$ to $-1$. In this case, Proposition~\ref{prop:power-func} proves that the variance of the unbiased estimator is exponentially larger than the one of the naive estimator. Additionally, the value of the naive estimator is in $\{-1, 1\}$ and thus its bias is bounded by 2. This proves that the MSE of the unbiased estimator is also exponentially larger than the one of the naive estimator.
\end{proof}

\paragraph{A Monotone Function with Exponential Variance Growth.}
As a second example we show that exponential variance growth can occur even for a monotone indicator function where the naive estimator has constant MSE.
To simplify calculations we focus on the regime $\varepsilon \le 1$.

\begin{proposition}\label{prop:threshold-exponential-lower}
Let $n$ be even and $f(x)=\mathbf{1}\{\sum_i x_i>0\}$. Let $g$ be from \Cref{thm:unbiase-multivariate}. If $p=e^{-\varepsilon}$ with $\varepsilon\le 1$, then for every $x$ with $\sum_i x_i=0$, $\Var{g(\tilde x)}=\exp(\Omega(n))$.
\end{proposition}

The proof of Proposition~\ref{prop:threshold-exponential-lower} can be found in Appendix~\ref{app:threshold-exponential-lower}. The idea is to show that for an all-zeros input the variance is lower bounded by the event that zero-sum noise is added by the Discrete Laplace mechanism, which happens with probability $\Omega(1/\sqrt{n})$, in which case the estimator has value $\exp(\Omega(n))$.

\section{Postprocessing Discrete Laplace Noise into Laplace or Staircase Noise}\label{sec:tranformation}

We show that discrete Laplace noise can be converted, by input-independent postprocessing, into either continuous Laplace noise or staircase noise at the same privacy level. The constructions are local to each unit interval and work by adding an independent random variable supported on $[-1,1]$.
It suffices to consider the univariate setting, since the multivariate setting just repeats the same method to each coordinate.

Fix $\varepsilon>0$ and write $p=e^{-\varepsilon}$. Let $\eta \sim \DLap(p)$, so that $\tilde{x} = x + \eta$ is the discrete Laplace mechanism at privacy level $\varepsilon$ for unit sensitivity. Since all transformations below are postprocessings of $\eta$, they preserve $\varepsilon$-DP automatically.
Both our results use the following lemma.

\begin{lemma}\label{lem:cell-convolution}
    Let $q$ be a probability distribution with support contained in $[-1,1]$, symmetric around zero, and let $Y\sim q$ be independent of $\eta$.
    For $z \in \mathbb R$, write $|z|=n+u$ with $n \in \mathbb Z_{\ge 0}$ and $u\in[0,1)$. Then $\eta+Y$ has density function
    \[
        f_{\eta+Y}(z)=\tfrac{1-p}{1+p}\,p^n\bigl(q(u)+p\,q(u-1)\bigr) \enspace .
    \]
\end{lemma}

\begin{proof}
    First suppose $z\geq 0$.
    Since $q$ is supported on $[-1,1]$, only the atoms $n$ and $n+1$ can contribute to $z=n+u$. Thus
    \[
        f_{\eta+Y}(n+u) = \Pr[\eta=n]\,q(u)+\Pr[\eta=n+1]\,q(u-1) = \tfrac{1-p}{1+p}\,p^n\bigl(q(u)+p\,q(u-1)\bigr) \enspace .
    \]
    Here we use that $q(u-1)=q(1-u)$ since $q$ is symmetric around zero.
    For $z < 0$, by symmetry the same formula with $-z = n + u$ determines the density.
\end{proof}

Now we show a method for post-processing from the discrete Laplace mechanism to the continuous Laplace mechanism while preserving the same privacy level $\varepsilon = \ln(1/p)$. 

\begin{theorem}\label{thm:laplace}
    Let $Y$ be independent of $\eta$ with density
    \[
        q_{\mathrm{Lap}}(y) = 
        \begin{cases}
        \tfrac{\ln(1/p)}{2(1-p)^2}\Bigl(p^{|y|}-p^{2-|y|}\Bigr) & \text{for } |y|\le 1\\
        0 & \text{otherwise}
        \end{cases} \enspace .
    \]
    Then $\eta+Y$ has the Laplace density
    \(
        f(z)=\frac{\ln(1/p)}{2} p^{|z|}.
    \)
\end{theorem}

\begin{proof}
    It is easy to verify that $q_{\mathrm{Lap}}$ is a density function, is symmetric around zero, and has support $[-1,1]$.
    Fix $z$ and write $|z|=n+u$ with $n \ge 0$ and $u \in[0,1)$. 
    We have
    \[
        q_{\mathrm{Lap}}(u)+p\,q_{\mathrm{Lap}}(1-u) 
            = \tfrac{\ln(1/p)}{2(1-p)^2} \Bigl((p^u-p^{2-u})+p(p^{1-u}-p^{1+u})\Bigr) 
            = \tfrac{(1+p)\ln(1/p)}{2(1-p)}\,p^u \enspace .
    \]
    By Lemma~\ref{lem:cell-convolution} and using $q_{\mathrm{Lap}}(u-1)=q_{\mathrm{Lap}}(1-u)$:
    \[
        f_{\eta+Y}(z) 
        = \tfrac{1-p}{1+p}p^n\bigl(q_{\mathrm{Lap}}(u)+p\,q_{\mathrm{Lap}}(u-1)\bigr)
            = \tfrac{1-p}{1+p}p^n\cdot \tfrac{(1+p)\ln(1/p)}{2(1-p)}\,p^u
            = \tfrac{\ln(1/p)}{2}p^{|z|} \enspace. \qedhere
    \]
\end{proof}

\Cref{{fig:dlap-to-laplace}} illustrates the transformation.
We defer the method and the accompanying proofs for post-processing to the Staircase mechanism to Appendix~\ref{app:starcase-transformation}.

\paragraph{Discussion.} We have shown in this section that both the continuous Laplace mechanism and the Staircase mechanism can be simulated by post-processing a draw from the discrete Laplace mechanism. Consequently, any unbiased estimator for either of these mechanisms induces an unbiased estimator based on discrete Laplace output. For example, if $g$ is an unbiased estimator for the continuous Laplace mechanism and $Y$ is independent post-processing noise from \Cref{thm:laplace}, then $y \mapsto g(y+Y)$, evaluated on discrete Laplace output, has the same distribution as $g$ evaluated on continuous Laplace output. The same argument applies to the Staircase mechanism. By \Cref{thm:uniqueness}, the discrete Laplace unbiased estimator has minimum variance among all unbiased estimators, including randomized ones. Hence, no unbiased estimator based on the continuous Laplace or Staircase mechanism can have smaller variance than the corresponding unbiased estimator based directly on the discrete Laplace mechanism. This fact, along with the possibility to transform the output into a continuous Laplace or Staircase distribution if necessary, means that the discrete Laplace mechanism should be preferred over the Laplace and Staircase mechanisms whenever possible.

\section{Unbiased Estimators for Specific Multivariate Functions}\label{sec:faster}

The computational bottleneck of Theorem~\ref{thm:unbiase-multivariate} is due to the exponential size of the set $B_y = \{y + \xi : \xi \in \{-1, 0, 1\}^n\}$ on which $f$ needs to be evaluated. 
Fortunately, many common functions of interest exhibit structure which makes it tractable to compute $g(y)$. In this section, we give polynomial time algorithms for unbiasing the min, max, general order statistics, a decision tree evaluation, the entropy function, and multivariate polynomials. These functions all satisfy the condition of \Cref{thm:unbiase-multivariate} because they exhibit polynomial growth (formalized in Lemma~\ref{lem:poly}). These results are unified by a general technique which can give efficient unbiased estimators for any function $f$ that is \emph{locally decomposable} as a sum of polynomially many ``easy'' functions on the domain $B_y$, for any $y \in \Z^n$.
Easy functions refer to a basis class of functions where it is possible to evaluate~\eqref{eq:unbiased-multivariate} efficiently. Due to space constraints, we provide the details of this in Appendix~\ref{app:faster}.

\noindent \textbf{Remarks on Running Time and Error.} We measure running time in the real-RAM arithmetic model; specifically, we assume that the input $y$ is given as a sequence of $n$ integers, we can evaluate $f$ at any point in constant time (producing a real number), and we can do real-valued arithmetic in constant time in addition to the standard operations of a Turing machine. In preliminary experiments, we observed that the unbiased estimators in this section can have high error that is possibly much higher than that of the naive estimator (consistent with the gap in Theorem~\ref{thm:l2-bound-multivariate}). Thus, the results in this section are primarily of theoretical interest, and further work is necessary to demonstrate their practical performance.

\paragraph{Order Statistics}
For an integer $1 \leq i \leq n$, the $i$th order statistic of an input $x_1, \ldots, x_n$ is the $i$th element in the sorted input. This encapsulates the minimum, median, and maximum. The unbiased estimators for order statistics have the following running times:

\begin{theorem} (Informal statement of Theorems~\ref{thm:unbiased-ex}, \ref{thm:unbiased-ord}).
    The estimator~(\ref{eq:unbiased-multivariate}) for min and max can be computed in $O(n)$ time. For the $i$th order statistic (e.g., median) it can be computed in $O(n^2)$ time.
\end{theorem}

\paragraph{Decision Tree Evaluation}
A decision tree evaluation function can be described by a tree $T$. Each internal node of the tree is labeled with a predicate of the form $x_i \leq c$ with $c$ a real-valued constant. Internal nodes have two children corresponding to YES and NO. Finally, the leaves are labeled with real-valued constants. The function is evaluated at $x_1, \ldots, x_n$ by following the valid path from root to leaf, and then outputting the value at the leaf. We show:

\begin{theorem} (Informal statement of Theorem~\ref{thm:ub-dt})
The unbiased estimator of a decision tree evaluation function with $s$ internal nodes can be computed in $O(ns)$ time.
\end{theorem}

\paragraph{Entropy and KL Divergence} If $x_1, \ldots, x_n$ is interpreted as a vector of counts over a domain of size $n$, then it is natural to consider the entropy of the empirical probability distribution $p_1, \ldots, p_n$ where $p_i = \frac{x_i}{\sum_j x_j}$. The entropy takes the form $H(x_1, \ldots, x_n) = \sum_{i=1}^n h(\frac{x_i}{\sum_j x_j})$, where $h(z) = -z \log_2(z)$. The estimator for the entropy has the following runtime:

\begin{theorem}
    (Informal statement of Theorem~\ref{thm:entropy}) The unbiased estimator for the entropy can be computed in $O(n^2)$ time.
\end{theorem}
In \Cref{app:kl}, we extend these results to KL divergence and obtain an unbiased estimator with $O(n^3)$ running time.

\paragraph{Multivariate Polynomials} A multivariate monomial of the form $x_1^{a_1}\ldots x_n^{a_n}$ does not fit easily into the local decomposition framework. Nonetheless, it is possible to compute the unbiased estimator efficiently because the variables in the monomial interact through a product, and this allows one to factor~\eqref{eq:unbiased-multivariate} directly. Specifically, products of functions can be unbiased as the product of the unbiased estimators as follows:

\begin{theorem}\label{thm:unbiased-prod}
    Suppose $f(y_1, \ldots, y_n) = f_1(Y_1)f_2(Y_2)$ for two functions $f_1 : \Z^{n_1} \rightarrow \R$ and $f_2 : \Z^{n_2} \rightarrow \R$, where $Y_1 \sqcup Y_2$ is a partition of the variables $y_1, \ldots, y_n$. Let $g_1(Y_1), g_2(Y_2)$ be the unbiased estimators given in Theorem~\ref{thm:unbiase-multivariate}. Then the unbiased estimator $g(y_1, \ldots, y_n)$ for $f$ is equal to $g_1(Y_1) g_2(Y_2)$.
\end{theorem}
Thus, a multivariate monomial can be unbiased by multiplying the unbiased estimators for $x_i^{a_i}$ separately (which are given in Theorem~\ref{thm:unbiase-univariate}). A general multivariate polynomial can be unbiased by adding the estimators for each monomial.

\section{Empirical Results} \label{sec:Empirical}
Here we briefly discuss our empirical validation\footnote{The code used for those experiments is available in the following repository: \url{https://github.com/Gericko/Postprocessing-Discrete-Laplace}}, for which further details can be found in Appendix~\ref{appendix:empirical-results}.
\Cref{sec:estimation-polynomials} considers polynomials, for a use case in graph counting ($k$-star counting), where we emphasize the advantage of the discrete Laplace mechanism over the classical Laplace mechanism for high values of $\varepsilon$.
The $k$-stars experiments suggest that both the use of the discrete Laplace mechanism and the debiasing matter.

\Cref{sec:estimation-bounded-entropy} concerns empirical entropy, where we see a significantly improved root mean square error for low- and medium values of $\varepsilon$.
The entropy experiment, whose results are shown in \Cref{fig:bow-nips-main}, highlights that when the naive estimator consists of a sum of biased terms, debiasing yields a significant improvement when applied to a sparse bag-of-words dataset. 

Other functions, such as partition functions, also fall into the category of functions written as a sum over disjoint variables.
We consider empirical results on estimating partition functions in \Cref{sec:partition-function}, for a setting in which the naive estimator has considerable bias.

Finally, \Cref{sec:profile-estimation} provides the first practical demonstration of post-processing-based profile estimation.
One of the results is shown in \Cref{fig:exac-main}.
These results show that debiasing can help prevent the smoothing effect of the noise addition.
Overall, debiasing is most compelling when bias would otherwise accumulate.
Our empirical results show that even in non-distributed settings, there are functions for which our unbiased estimator leads to lower error than the naive plug-in estimator.
In particular, they confirm the theoretical contributions and the applicability of the proposed method.

\begin{figure}[h]
\centering
\begin{subfigure}{0.48\textwidth}
    \centering
    \includegraphics[width=\textwidth, height=4.2cm,keepaspectratio]{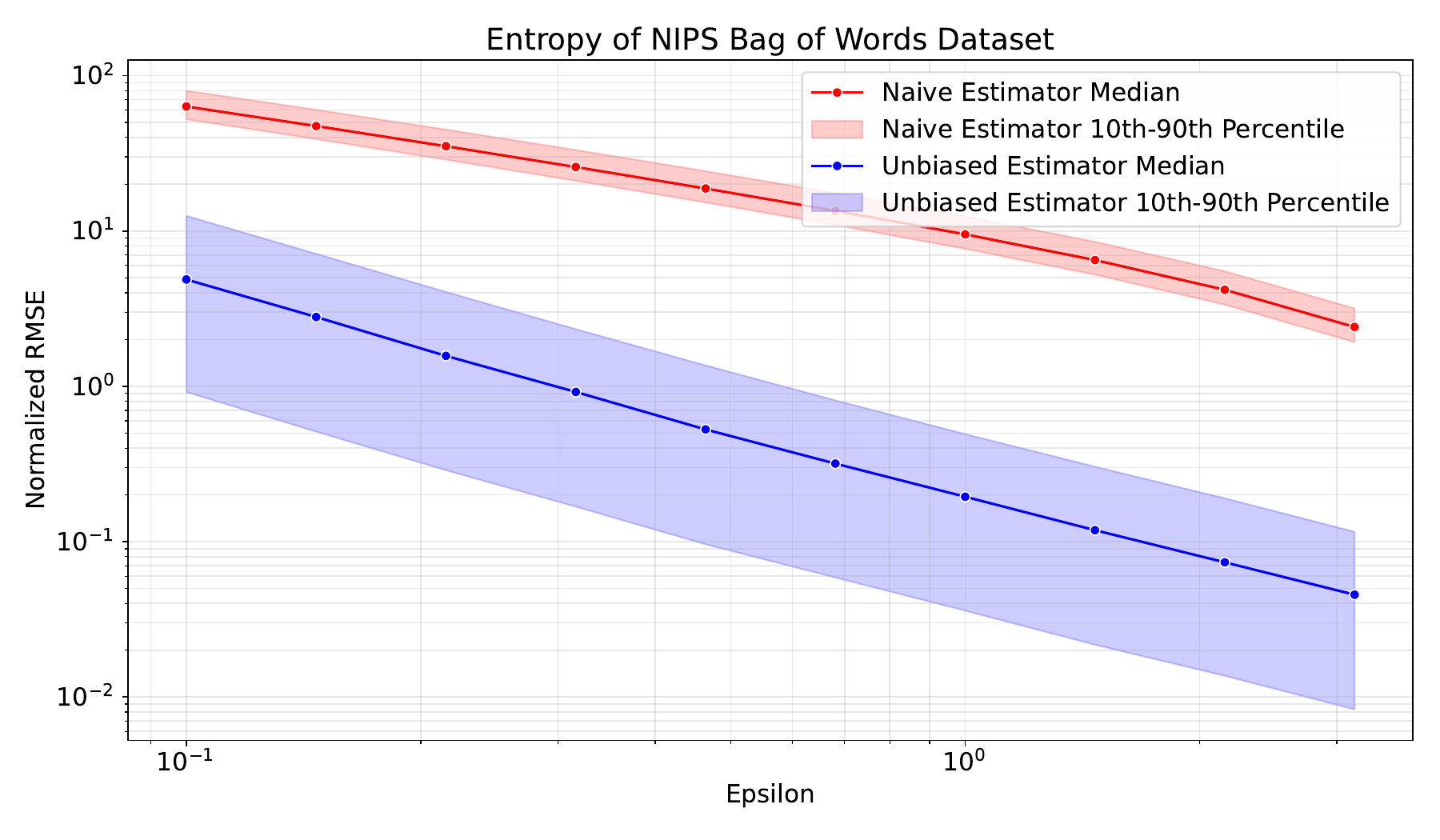}
    \caption{Root mean square error of the entropy estimation on the bag of words of NIPS articles. Quantiles are shown over about 30,000 runs for each $\varepsilon$ (20 repetitions over 1491 histograms). 
    }
    \label{fig:bow-nips-main}
\end{subfigure}\hfill
\begin{subfigure}{0.48\textwidth}
    \centering
    \includegraphics[width=\textwidth, height=4.2cm,keepaspectratio]{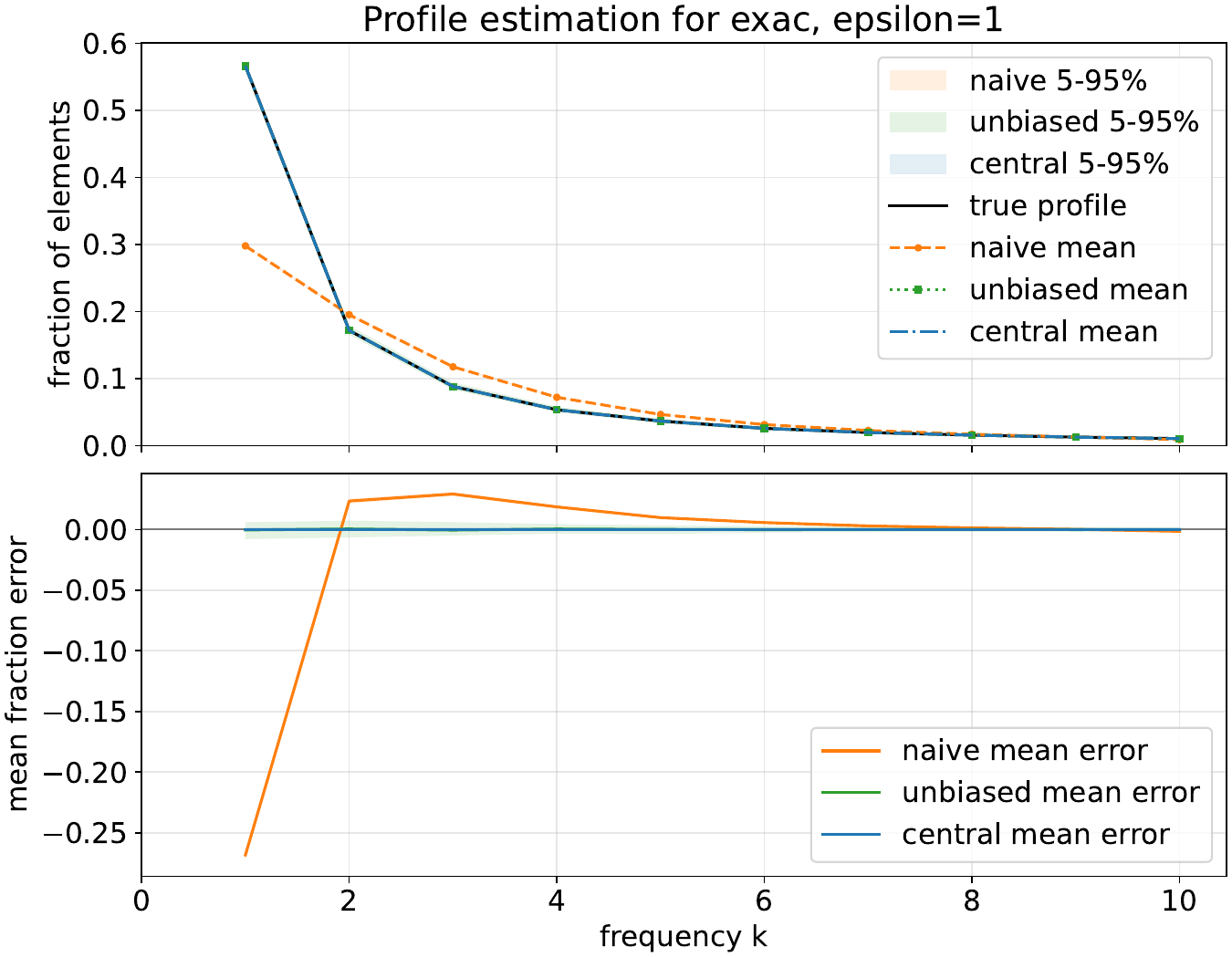}
    \caption{Profile estimation on ExAC histogram. Quantiles are shown for $\varepsilon=1$ and for over 200 noise realizations for each value of $k$.}
    \label{fig:exac-main}
\end{subfigure}
\caption{Examples where our estimator compares favorably to the naive estimator on real datasets.}
\label{fig:empirical-main}
\end{figure}

\section{Discussion and Limitations} \label{sec:discussion}
Our proposed method is general and imposes minimal assumptions on the properties of $f$.
Compared with previous work, this generality allows unbiased estimators to be constructed for a much wider class of functions.
One cost of this generality is computational: in general computing our estimator requires exponential time.
For many natural functions we show, however, that the estimator can be computed in polynomial time. A natural open problem of characterizing the classes of functions that admit polynomial-time unbiased estimation algorithms. 

A second limitation of our estimator is that it sometimes has substantially higher variance than the naive estimator.
This is the case for some of the functions we have studied, such as order statistics.
However, as we have seen in our theoretical and empirical results there are other cases of interest where the mean square error of our estimator is small or even significantly better than that of the naive estimator.
Characterizing the functions and inputs for which such an improvement occurs is another interesting avenue for future work.
Finally, unbiased estimators sometimes have worse error than estimators with a small, controlled bias --- it would be interesting to understand if such biased estimators can lead to better estimation of sums in distributed/federated settings.

\newpage

\bibliography{references.bib}

\newpage

\appendix
\section{Omitted Proofs for \texorpdfstring{\Cref{sec:DebiasDLap}}{Section 2}} \label{app:debias}

\subsection{Proof of \texorpdfstring{\Cref{thm:unbiase-univariate}}{Theorem 1}}
\begin{proof}
    Let $f: \Z \to \R$, $p \in (0,1)$, $x \in \Z$, $\tx = x + \DLap(p)$, and $g$ as defined in \Cref{eq:unbiased}.

   \begin{align*}
        \E{g(\tx)}
        = & \sum\limits_{i=-\infty}^{+\infty} \proba{\tx - x = i} g(x+i) \\
        = & \frac{1 - p}{1 + p} \sum\limits_{i=-\infty}^{+\infty} p^{|i|} g(x+i) \\
        = & \frac{1 - p}{1 + p} \sum\limits_{i=-\infty}^{+\infty} p^{|i|} \left( f(x+i) - \frac{p}{(1 - p)^2} (f(x+i+1) - 2 f(x+i) + f(x+i-1)) \right) \\
        = & \frac{(1 - p)^2 + 2p}{(1 + p)(1 - p)} \sum\limits_{i=-\infty}^{+\infty} p^{|i|} f(x+i) - \frac{p}{(1 + p)(1 - p)} \left[ \sum\limits_{i=-\infty}^{-1} p^{-i} f(x+i+1) \right. \\
          & \left. + \sum\limits_{i=0}^{+\infty} p^{i} f(x+i+1) + \sum\limits_{i=-\infty}^{0} p^{-i} f(x+i-1) + \sum\limits_{i=1}^{+\infty} p^{i} f(x+i-1) \right] \\
        = & \frac{1 + p^2}{(1 + p)(1 - p)} \sum\limits_{i=-\infty}^{+\infty} p^{|i|} f(x+i) - \frac{p}{(1 + p)(1 - p)} \\
          & \left[ \sum\limits_{i=-\infty}^{0} p^{-i+1} f(x+i) + \sum\limits_{i=1}^{+\infty} p^{i-1} f(x+i) + \sum\limits_{i=-\infty}^{-1} p^{-i-1} f(x+i) + \sum\limits_{i=0}^{+\infty} p^{i+1} f(x+i) \right] \\
        = & \frac{1 + p^2}{(1 + p)(1 - p)} \sum\limits_{i=-\infty}^{+\infty} p^{|i|} f(x+i) - \frac{p}{(1 + p)(1 - p)} \\
          & \left[ p \sum\limits_{i=-\infty}^{+\infty} p^{|i|} f(x+i) + p f(x) + \frac{1}{p} \sum\limits_{i=-\infty}^{+\infty} p^{|i|} f(x+i) - \frac{1}{p} f(x) \right] \\
        = & f(x)\qedhere
    \end{align*}
\end{proof}

\subsection{Proof of \texorpdfstring{\Cref{thm:unbiase-multivariate}}{Theorem 2}}
\begin{proof}
    Let $f: \Z^n \to \R$, $p \in (0,1)$, $(x_1, \ldots, x_n) \in \Z^n$, for all $i \in [n]$, $\tx_i = x_i + \DLap(p)$, and $g$ as defined in \Cref{eq:unbiased-multivariate}.

    \begin{align*}
        & \E[\tx_1, \ldots, \tx_n]{g(\tx_1, \ldots, \tx_n))} \\
        & = \E[\tx_1, \ldots, \tx_n]{\sum\limits_{(\xi_1, \ldots, \xi_n) \in \{-1, 0, 1\}^n} f(\tx_1 + \xi_1, \ldots, \tx_n + \xi_n) \prod\limits_{j=1}^n \alpha_{\xi_j}} \\
        & = \E[\tx_1, \ldots, \tx_{n-1}]{\E[\tx_n]{\sum\limits_{(\xi_1, \ldots, \xi_n) \in \{-1, 0, 1\}^n} f(\tx_1 + \xi_1, \ldots, \tx_n + \xi_n) \prod\limits_{j=1}^n \alpha_{\xi_j}}} \\
        & = \mathbb{E}_{\tx_1, \ldots, \tx_{n-1}} \left[ \sum\limits_{(\xi_1, \ldots, \xi_{n-1}) \in \{-1, 0, 1\}^{n-1}} \prod\limits_{j=1}^{n-1} \alpha_{\xi_j} \cdot \mathbb{E}_{\tx_n} \left[
        - \frac{p}{(1-p)^2} f(\tx_1 + \xi_1, \ldots, \tx_{n-1} + \xi_{n-1}, \tx_n - 1) \right.\right. \\
        & \left.\left. + \left(1 + \frac{2 p}{(1-p)^2}\right) f(\tx_1 + \xi_1, \ldots, \tx_{n-1} + \xi_{n-1}, \tx_n) - \frac{p}{(1-p)^2} f(\tx_1 + \xi_1, \ldots, \tx_{n-1} + \xi_{n-1}, \tx_n + 1)
        \right]\right]\\
        & = \E[\tx_1, \ldots, \tx_{n-1}]{\sum\limits_{(\xi_1, \ldots, \xi_{n-1}) \in \{-1, 0, 1\}^{n-1}} f(\tx_1 + \xi_1, \ldots, \tx_{n-1} + \xi_{n-1}, x_n) \prod\limits_{j=1}^{n-1} \alpha_{\xi_j}} \\
        & = f(x_1, \ldots, x_n)
    \end{align*}
    
    The second equality is derived using the independence of the random variables $\tx_i$.
    For the fourth equality, we apply \Cref{thm:unbiase-univariate} to the functions $g(x) = f(\tx_1 + \xi_1, \ldots, \tx_{n-1} + \xi_{n-1}, x)$ for each assignment of $(\xi_1, \ldots, \xi_{n-1})$.
    Finally the last equality is obtained by induction on $n$ by considering this time the function $g(x_1, \ldots, x_{n-1}) = f(x_1, \ldots, x_{n-1}, x_n)$.
\end{proof}

\section{Omitted Proofs for \texorpdfstring{\Cref{sec:properties}}{Section 3}} \label{app:bounding-l2-error}

\subsection{Proof of \texorpdfstring{\Cref{thm:uniqueness}}{Theorem 4}}\label{app:uniquenessproof}

    For $y \in \Z^n$ we will prove that if $g$ is an unbiased estimator of the function $y \mapsto 0$, then for all $y \in \Z^n, g(y) = 0$. From there, we can easily prove the theorem by noticing that if $g_1$ and $g_2$ are unbiased estimators of the same function $f$, then $g_1 - g_2$ is an unbiased estimator of the constant function $y \mapsto 0$.
    Let $p \in (0,1)$, $x \in \Z^n$, $\tx = x + \DLap(p)$, and $h$ be a deterministic unbiased estimator for $y \mapsto 0$.
    By the definition of the unbiased estimator, we have
    \begin{equation*}
        \E{h(\tx)} = \left(\tfrac{1 - p}{1 + p}\right)^n \sum_{i \in \Z^n} p^{\|i-x\|_1} h(i) = 0.
    \end{equation*}
    
    To obtain a cancellation in the sum, we will use the following equality.
    \[
        p^{|z-1|} + p^{|z+1|} - \tfrac{p^2 + 1}{p} p^{|z|} =
        \begin{cases}
            p - \frac{1}{p},    & z=0 \\[1.2ex]
            0,                  & \text{otherwise}
        \end{cases}
    \]

    For any fixed $x \in \Z^n$, by applying the unbiased estimator to $\tx - e_1$, $\tx + e_1$, and $\tx$, we have that the following linear combination is also equal to 0:
    \begin{align*}
        & \E{h(\tx - e_1)} + \E{h(\tx + e_1)} - \tfrac{p^2 + 1}{p} \E{h(\tx)} \\
        & = \left(\tfrac{1 - p}{1 + p}\right)^n \sum_{(i_2, \ldots, i_n) \in \Z^{n-1}} p^{\sum_{j=2}^n |i_j - x_j|} \sum_{i_1=-\infty}^{+\infty} \left[p^{|i_1 - x_1 - 1|} + p^{|i_1 - x_1 + 1|} - \tfrac{p^2 + 1}{p} p^{|i_1 - x_1|}\right] h(i) \\
        & = \left(\tfrac{1 - p}{1 + p}\right)^n \left(p - \tfrac{1}{p}\right) \sum_{(i_2, \ldots, i_n) \in \Z^{n-1}} p^{\sum_{j=2}^n |i_j - x_j|} h(x_1, i_2, \ldots, i_n) \\
        & = \tfrac{1 - p}{1 + p} \left(p - \tfrac{1}{p}\right) \E{h(x_1, \tx_2, \ldots, \tx_n)}
    \end{align*}

    This shows that the zero expectation is preserved when replacing the random variable $\tx_1$ by $x_1$.
    By performing the same cancellation for each coordinate $j$, we obtain $h(x) = 0$. 
    This is true for every $x \in \Z^n$ and concludes the proof of the uniqueness of the deterministic estimator.

    Concerning the optimality, let $U$ be auxiliary randomness and $\delta(\tx, U)$ an unbiased estimator of $f(x)$. Using \Cref{thm:rao}, we have that $\bar{\delta}(\tx)$ is a deterministic unbiased estimator of $f(x)$ with smaller variance than $\delta(\tx, U)$. However, by the uniqueness of the deterministic unbiased estimator, $\bar\delta$ is~$g$. \qed

\subsection{Proof of \texorpdfstring{\Cref{thm:l2-bound-multivariate}}{Theorem 5}} \label{app:proof of l2-bound-multivariate}
For simplicity, we first consider the $n=1$ case: 
First, we state and prove the following lemma: 
\begin{lemma}
    For $f: \Z \to X$, $p \in (0,1)$, $x \in \Z$, $\tx = x + \DLap(p)$, and $\xi \in \{-1, +1\}$,
    \[
        \E{(f(\tx + \xi) - f(x))^2} \leq \frac{1}{p} \cdot \MSE[f]{x}.
    \]
\end{lemma}
\begin{proof}
    \begin{align*}
        \E{(f(\tx + \xi) - f(x))^2}
        & = \frac{1 - p}{1 + p} \sum\limits_{i=-\infty}^{+\infty} p^{|i|} (f(x+i + \xi) - f(x))^2 \\
        & = \frac{1 - p}{1 + p} \sum\limits_{j=-\infty}^{+\infty} p^{|j - \xi|} (f(x+j) - f(x))^2 \\
        & \leq \frac{1 - p}{1 + p} \sum\limits_{j=-\infty}^{+\infty} \frac{p^{|j|}}{p} (f(x+j) - f(x))^2 \\
        & = \frac{1}{p} \cdot \MSE[f]{x}\qedhere
    \end{align*}
\end{proof}
Then, we prove the theorem for $n=1$ by applying the lemma above:
\begin{proof}
    \begin{align*}
        \MSE[g]{x} 
        & = \E{\left(f(\tx) - \frac{p}{(1-p)^2} (f(\tx+1) - 2 f(\tx) + f(\tx-1)) - f(x)\right)^2} \\
        & = \frac{1}{(1-p)^4} \E{\left(\left(1 + p^2\right)\left(f(\tx) - f(x)\right) - p \left(f(\tx+1) - f(x)\right) - p \left(f(\tx-1) - f(x)\right)\right)^2} \\
        & \leq \frac{3}{(1-p)^4} \E{\left(1 + p^2\right)^2\left(f(\tx) - f(x)\right)^2 + p^2 \left(f(\tx+1) - f(x)\right)^2 + p^2 \left(f(\tx-1) - f(x)\right)^2} \\
        & \leq \frac{3}{(1-p)^4} \left(\left(1 + p^2\right)^2 \MSE[f]{x} + p \cdot \MSE[f]{x} + p \cdot \MSE[f]{x}\right) \\
        & = 3 \frac{\left(1 + p^2\right)^2 + 2p}{(1-p)^4} \cdot \MSE[f]{x}.\qedhere
    \end{align*}
\end{proof}

Now we consider the general case.
Before proving the main theorem, we state and prove the following lemma (similar to the $n=1$ case):
\begin{lemma}
    For $f: \Z^n \to X$, $p \in (0,1)$, $(x_1, \ldots, x_n) \in \Z^n$, and for all $i \in [n]$, $\tx_i = x_i + \DLap(p)$, and $(\xi_1, \ldots, \xi_n) \in \{-1, 0, 1\}^n$,
    \[
        \E{(f(\tx_1 + \xi_1, \ldots, \tx_n + \xi_n) - f(x_1, \ldots, x_n))^2} \leq \frac{1}{p^{\sum_{j=1}^n |\xi_j|}} \cdot \MSE[f]{x_1, \ldots, x_n}.
    \]
\end{lemma}
\begin{proof}
    \begin{align*}
        & \E{(f(\tx_1 + \xi_1, \ldots, \tx_n + \xi_n) - f(x_1, \ldots, x_n))^2} \\
        & = \left(\frac{1 - p}{1 + p}\right)^n \sum\limits_{(i_1, \ldots, i_n) \in \Z^n} p^{\sum_{j=1}^n |i_j|} (f(x_1 + \xi_1 + i_1, \ldots, x_n + \xi_n + i_n) - f(x_1, \ldots, x_n))^2 \\
        & = \left(\frac{1 - p}{1 + p}\right)^n \sum\limits_{(i_1, \ldots, i_n) \in \Z^n} p^{\sum_{j=1}^n |i_j - \xi_j|} (f(x_1 + i_1, \ldots, x_n + i_n) - f(x_1, \ldots, x_n))^2 \\
        & \leq \left(\frac{1 - p}{1 + p}\right)^n \sum\limits_{(i_1, \ldots, i_n) \in \Z^n} p^{\sum_{j=1}^n |i_j| - |\xi_j|} (f(x_1 + i_1, \ldots, x_n + i_n) - f(x_1, \ldots, x_n))^2 \\
        & = \frac{1}{p^{\sum_{j=1}^n |\xi_j|}} \cdot \E{(f(\tx_1, \ldots, \tx_n) - f(x_1, \ldots, x_n))^2}.\qedhere
    \end{align*}
\end{proof}

Then, we prove \Cref{thm:l2-bound-multivariate} by applying the lemma above:
\begin{proof}
    \begin{align*}
        & \MSE[g]{x_1, \ldots, x_n} \\
        = \ & \E{\left(\sum\limits_{(\xi_1, \ldots, \xi_n) \in \{-1, 0, 1\}^n} f(\tx_1 + \xi_1, \ldots, \tx_n + \xi_n) \prod\limits_{j=1}^n \alpha_{\xi_j} - f(x_1, \ldots, x_n)\right)^2} \\
        \leq \ & 3^n \sum\limits_{(\xi_1, \ldots, \xi_n) \in \{-1, 0, 1\}^n} \E{(f(\tx_1 + \xi_1, \ldots, \tx_n + \xi_n) - f(x_1, \ldots, x_n))^2} \prod\limits_{j=1}^n \alpha_{\xi_j}^2 \\
        \leq \ & 3^n \sum\limits_{(\xi_1, \ldots, \xi_n) \in \{-1, 0, 1\}^n} \MSE[f]{x_1, \ldots, x_n} \prod\limits_{j=1}^n \frac{\alpha_{\xi_j}^2}{p^{|\xi_j|}} \\
        = \ & \left( 3 \frac{\left(1 + p^2\right)^2 + 2p}{(1-p)^4} \right)^n \MSE[f]{x_1, \ldots, x_n}.\qedhere
    \end{align*}
\end{proof}

\subsection{Proof of Proposition~\texorpdfstring{\ref{prop:power-func}}{6}}\label{app:proof-power-func}

\begin{proof}
    Let $p \in (0,1)$, $\gamma \in (-1/p,-p) \cup (p, 1/p)$, $x \in \Z^n$, and $\tx = x + \DLap(p)$.
    In one dimension, 
    \begin{align*}
        \E{|\gamma^{\tx}|}
        & = \tfrac{1 - p}{1 + p} \sum_{i=-\infty}^{+\infty} p^{|i|}|\gamma|^{x + i}
         = \tfrac{1 - p}{1 + p} |\gamma|^{x} \left( \sum_{i=1}^{+\infty} \left(\frac{p}{|\gamma|}\right)^i + \sum_{i=0}^{+\infty} (p |\gamma|)^i \right).
    \end{align*}
    As $|\gamma| \in (p, 1/p)$, both $p / |\gamma| < 1$ and $p|\gamma| < 1$. Thus the expected value is finite. This carries to the multidimensional case as the noise added to each component are independent and thus
    \(
        \E{|f_{\gamma}(\tx)|} = \prod_{i=1}^n \E{|\gamma^{\tx_i}|}.
    \)

    We use the same independence property to get the unbiased estimator.
    It is useful to express this using the operator form of the estimator. Let
    \[
    \Delta_i f(y) = f(y+e_i)-2f(y)+f(y-e_i)
    \]
    denote the discrete second derivative in coordinate $i$, and write
    \(T_i = I - c\Delta_i\), \(c = \frac{p}{(1-p)^2}\).
    Then \Cref{thm:unbiase-multivariate} can be written compactly as
    \[
    g = \prod_{i=1}^n T_i f,
    \]
    where the operators commute because they act on different coordinates.
    As $\Delta_i(f_{\gamma}(y)) = \frac{(1 - \gamma)^2}{\gamma} f_{\gamma}(y)$, this gives $g_{\gamma}(y) = A(\gamma)^n f_{\gamma}(y)$.

    Last, whenever the variance of $f_{\gamma}(\tx)$ exists, it follows immediately that the variance of $g_{\gamma}(\tx)$ exists as well and that $\Var{g_{\gamma}(\tx)} = A(\gamma)^{2n} \Var{f_{\gamma}(\tx)}$. The existence of that variance is given by the previous claim on expected value by noticing that $\Var{f_{\gamma}(\tx)}$ converges exactly when $\E{f_{\gamma^2}(\tx)}$ converges.
\end{proof}

\subsection{Proof of Proposition~\ref{prop:threshold-exponential-lower}}\label{app:threshold-exponential-lower}

\begin{proof}
Write $\alpha_0=1+2c$, $\alpha_{+1}=\alpha_{-1}=-c$, where $c=p/(1-p)^2$. For $y$ with $\sum_i y_i=0$,
\[
g(y)=\sum_{\xi:\,\sum_i\xi_i>0}\prod_i \alpha_{\xi_i}.
\]
Let $\beta_k=\sum_{\xi:\,\sum_i\xi_i=k}\prod_i\alpha_{\xi_i}$. Then $\beta_k=\beta_{-k}$ and $\sum_k\beta_k=1$, so $g(y)=\sum_{k>0}\beta_k=\tfrac{1-\beta_0}{2}$.
Since all terms in the sum defining $\beta_0$ are non-negative we can lower bound $\beta_0$ using vectors $\xi$ with $n/2$ entries $+1$ and $n/2$ entries $-1$:
\[
\beta_0 \ge \binom{n}{n/2}c^n \ge \tfrac{(2c)^n}{n+1}.
\]
For $\varepsilon \le 1$ we have $p \ge e^{-1}$ and thus $c = p/(1-p)^2> 0.92$.
Consider the event $E$ that $\sum_i \xi_i=0$. 
By standard concentration bounds, $\Pr[E]=\Theta(n^{-1/2})$. 
Conditioned on $E$, we have $\sum_i\tilde x_i=0$, so
\[
|g(\tilde x)|=\tfrac{|1-\beta_0|}{2}\ge \Omega\bigl((2c)^n/n\bigr).
\]
Since $\E{g(\tilde x)} = 0$,
\[
\Var{g(\tilde x)} \ge \Pr[E]\cdot \Omega\bigl((2c)^{2n}/n^2\bigr)
= \Omega\bigl((2c)^{2n}/n^{5/2}\bigr),
\]
which is exponential since $2c>1$.
\end{proof}

\section{Details on Transformations to Other Mechanisms} \label{app:transformation}

\subsection{Illustration of the Conversion to the Laplace Mechanism} 

\begin{figure}[h]
\centering
\includegraphics{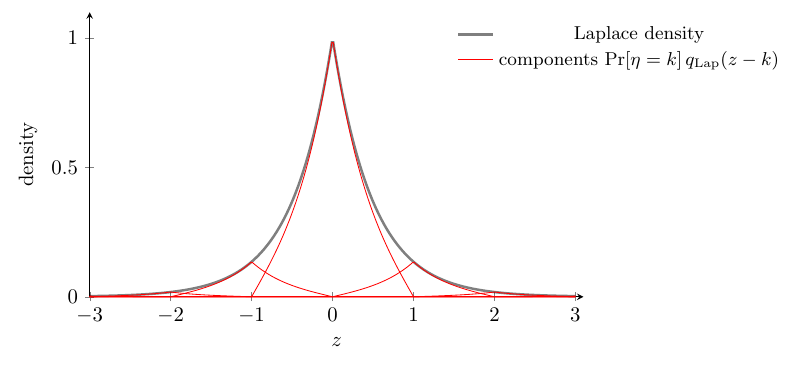}
\caption{
Illustration of Theorem~\ref{thm:laplace} for $\varepsilon=2$ (so $p=e^{-2}$).
Each red curve is one translated component
$\Pr[\eta=k]\,q_{\rm Lap}(z-k)$ corresponding to an atom
of the discrete Laplace distribution
$\eta\sim\mathrm{DLap}(e^{-\varepsilon})$.
Summing all components yields exactly the continuous Laplace density
$\frac{\varepsilon}{2}e^{-\varepsilon |z|}$ shown in black.
}
\label{fig:dlap-to-laplace}
\end{figure}

\subsection{Conversion to the Staircase Mechanism} \label{app:starcase-transformation}

We next treat the Staircase noise distribution~\cite{geng2014optimal,geng2015optimal} in the standard form used in the differential privacy literature.
Like for the Laplace mechanisms we use the parameter $p = e^{-\varepsilon}$ and further introduce a shape parameter $\gamma\in[0,1/2]$.
Define the symmetric density $f_{\gamma}$ by
\[
    f_{\gamma}(x)=
        \begin{cases}
            a_{\gamma}\,p^k & \text{ for } |x|\in [k,k+\gamma),\, k\in\mathbb Z_{\ge 0}\\
            a_{\gamma}\,p^{k+1} & \text{ for } |x|\in [k+\gamma,k+1),\, k\in\mathbb Z_{\ge 0}
        \end{cases}
\]
where
\(
    a_{\gamma}=\frac{1-p}{2(\gamma+p(1-\gamma))}
\).
Additive noise using this density function is the staircase mechanism with shape parameter $\gamma$, satisfying $\varepsilon$-differential privacy~\cite{geng2014optimal,geng2015optimal}.

\begin{theorem}\label{thm:staircase}
    Fix $\gamma\in[0,1/2]$. Let $\eta$ be sampled from the discrete Laplace distribution with parameter $p$, and let $Y_{\gamma}$ be independent of $\eta$ with density
    \[
        q_{\gamma}(y) =
            \begin{cases}
                \tfrac{1+p}{2(\gamma+p(1-\gamma))} & \text{for } |y|<\gamma\\
                \tfrac{p}{2(\gamma+p(1-\gamma))} & \text{for } \gamma\le |y|\le 1-\gamma\\
                0 & \text{otherwise }
            \end{cases} \enspace .
    \]
    Then $\eta+Y_{\gamma}$ has density $f_{\gamma}$.
\end{theorem}

\begin{proof}
    The function $q_{\gamma}$ is symmetric and nonnegative.
    Its total mass (the integral over all reals) is
    \[
        2\gamma\cdot \tfrac{1+p}{2(\gamma+p(1-\gamma))} + 2(1-2\gamma)\cdot \tfrac{p}{2(\gamma+p(1-\gamma))} = 1,
    \]
    so it is indeed a density. For any $z$ we write $|z|=n+u$ with $n\ge 0$ and $u\in[0,1)$, and apply Lemma~\ref{lem:cell-convolution}.    
    If $u<\gamma$, then $q_{\gamma}(u)=\frac{1+p}{2(\gamma+p(1-\gamma))}$ and $q_{\gamma}(u-1)=0$, because $|u-1|>1-\gamma$. Therefore
    \[
        f_{X+Y_{\gamma}}(z) 
            = \tfrac{1-p}{1+p}p^n\cdot \tfrac{1+p}{2(\gamma+p(1-\gamma))}
            = a_{\gamma}p^n.
    \]
    
    If $\gamma\le u<1-\gamma$, then $q_{\gamma}(u) = q_{\gamma}(u-1) = \frac{p}{2(\gamma+p(1-\gamma))}$, and hence
    \[
        f_{X+Y_{\gamma}}(z) 
            = \tfrac{1-p}{1+p}p^n\cdot \tfrac{p(1+p)}{2(\gamma+p(1-\gamma))}
            = a_{\gamma}p^{n+1}.
    \]
    
    If $1-\gamma\le u$, then $q_{\gamma}(u)=0$ and $q_{\gamma}(u-1)=\frac{1+p}{2(\gamma+p(1-\gamma))}$, so again
    \[
        f_{X+Y_{\gamma}}(z)
            = \tfrac{1-p}{1+p}p^n\cdot p\tfrac{1+p}{2(\gamma+p(1-\gamma))}
            = a_{\gamma}p^{n+1}.
    \]
    
    This matches exactly the definition of $f_{\gamma}$.
\end{proof}
The transformation is illustrated in \Cref{fig:dlap-to-staircase}.

\begin{figure}[t]
\centering
\includegraphics{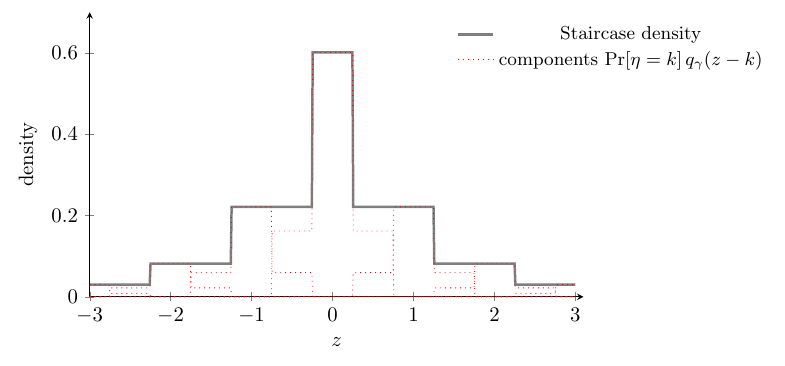}
\caption{
Illustration of Theorem~\ref{thm:staircase} for $\varepsilon=1$ and $\gamma=\tfrac14$
(so $p=e^{-1}$).
Each red dashed curve is one translated component
$\Pr[\eta=k]\,q_\gamma(z-k)$ corresponding to an atom
of the discrete Laplace distribution
$\eta\sim\mathrm{DLap}(e^{-\varepsilon})$.
Summing all components yields exactly the Staircase density shown in black.
}
\label{fig:dlap-to-staircase}
\end{figure}

\section{Omitted Details from \texorpdfstring{\Cref{sec:faster}}{Section 5}}\label{app:faster}

In Section~\ref{app:faster-formula}, we describe the general framework for obtaining faster unbiased formula. In the remaining sections, we use this framework to obtain faster estimators for various functions.

The functions in this section satisfy the condition of Theorem~\ref{thm:unbiase-multivariate}, and thus have unbiased estimators, because they all exhibit subexponential growth which ensures their expectation exists and is finite. We formalize this in the following lemma.

\begin{lemma}\label{lem:poly}
    For $p \in (0,1)$, suppose $|f(y)| \leq r(\|y\|_1)$, where $r(k) : \R \rightarrow \R$ is an increasing function such that $\lim_{k \rightarrow \infty}r(k)p^{k/4} = 0$. For a given $x \in \Z^d$, let $\tilde{x} = x + \DLap(p)^d$. Then, $\E{f(\tilde{x})}$ exists and is finite.
\end{lemma}
\begin{proof}
    We will show that the expectation converges to a finite value. To do so, observe the expectation can be written as a limit of the following partial sums:
    \begin{align*}
        \E{f(\tilde{x})} &= \sum_{z \in \mathbb{Z}^d} f(x+z) \left(\frac{1-p}{1+p}\right)^d p^{\|z\|_1} \\
        &= \lim_{k \rightarrow \infty} \sum_{z : \|z\|_1 \leq k} f(x+z) \left(\frac{1-p}{1+p}\right)^d p^{\|z\|_1}.
    \end{align*}
    Define the partial sum $P(k) = \sum_{z : \|z\|_1 \leq k} f(x+z) p^{\|z\|_1}$. We will show that for any $\alpha > 0$, there exists $k_0$ such that for all $k_1, k_2 \geq k_0$ such that $k_1 < k_2$, we have $|P(k_2) - P(k_1)| \leq \alpha$, and thus the above limit exists.

    To do this, we bound 
    \begin{align*}
        |P(k_2) - P(k_1)| &= \left| \sum_{k=k_1}^{k_2}\sum_{z : \|z\|_1 = k} f(x+z) p^k \right| \\
        &\leq 
        \sum_{k=k_1}^{k_2}\sum_{z : \|z\|_1 = k} |f(x+z)|p^k \\
        &\leq \sum_{k=k_1}^{k_2}\sum_{z : \|z\|_1 = k} p^{k} r\left(\|x+z\|_1  \right) \\
        &\leq \sum_{k=k_1}^{k_2} \sum_{z:\|z\|_1=k} p^{k} r\left( \|x\|_1 + \|z\|_1 \right) \\
        &= \sum_{k=k_1}^{k_2} \sum_{z:\|z\|_1=k} p^{k} r\left( \|x\|_1 + k \right) \\
        &\leq \sum_{k=k_1}^{k_2} (n+k)^n p^{k} r\left( \|x\|_1 + k \right),
    \end{align*}
    where the last step follows because $|\{z:\|z\|_1 = k\}\| = \binom{n+k-1}{n-1} \leq (n+k)^n$. $k_0$ can be made large enough so that $(n+k)^n p^{k/4} < 1$ for all $k \geq k_0$. 
    By the property of $r$, $k_0$ can be made large enough so that $r(\|x\|_1 + k) p^{k/2} \leq r(2k) p^{k/2} \leq \beta$ for all $k \geq k_0$. This means the above sum can be bounded as $\sum_{k=k_1}^\infty p^{k/4} \beta < \frac{\beta}{1-p^{1/4}}$, and the proof is finished by taking $\beta < \alpha (1-p^{1/4})$.
\end{proof}

\subsection{A General Debiasing Technique Via Local Decompositions}\label{app:faster-formula}
The unbiased estimator of Theorem~\ref{thm:unbiase-multivariate} depends only on the values of $f$ in the set $y+ \{-1, 0, 1\}^n$. Suppose that $f$ takes a constant value $C$ on a region $y + R$ where $R \subseteq \{-1, 0, 1\}^d$. Then, the total contribution of the region $y + R$ to the unbiased estimator $g$ is $C$ times
\[
    \sum_{\xi \in R} \prod_{j=1}^n \alpha_{\xi_j} \triangleq \vol(R),
\]
which can be interpreted as a type of volume of $R$.
Thus, if $f$ can be decomposed into a polynomial number of regions on the entire space $y + \{-1, 0, 1\}^n$ \emph{and} each region has an efficiently computable volume, then $g(y)$ can be efficiently computed. The precise notion of decomposability we need is:

\begin{definition}\label{def:decomp}
    Let $\mathcal{B}$ denote a set of subsets of $\{-1, 0, 1\}^n$. We will refer to such a $\mathcal{B}$ as a basis.
    We say a function $f : \Z^m \rightarrow \R$ is $m$-locally decomposable under a basis $\mathcal{B}$ if for any point $y \in \Z^n$, there exist sets $S_1, \ldots, S_m \in \mathcal{B}$ and coefficients $c_1, \ldots, c_m \in \R$ such that $f(y + \xi) = \sum_{i=1}^m c_i \mathbf{1}[\xi \in S_i]$ for all $\xi \in \{-1,0,1\}^n$.
\end{definition}

Observe that the decomposition may be different for each point $y$ and only needs to hold in a local neighborhood $y + \{-1, 0, 1\}^n$. We say that $f$ is \emph{efficiently} decomposable under $\mathcal{B}$ if the volume of each $S \in \mathcal{B}$ is efficiently computable, $f$ is $m$-locally decomposable with $m = \text{poly}(n)$, and the decomposition $S_1, \ldots, S_m, c_1, \ldots, c_m$ can be efficiently computed for any $y$. 

By combining~\Cref{eq:unbiased} with \Cref{def:decomp}, the unbiased estimator of $f$ can be written as
\begin{equation}\label{eq:unbias-rewrite}
    g(y) = \sum_{i=1}^m c_i \vol(S_i),
\end{equation}
and this equation is efficiently computable as $f$ is efficiently decomposable under $\mathcal{B}$.

A particularly nice basis is the set $\mathcal{R}_n$ of $n$-dimensional rectangles. A rectangle in $\{-1, 0, 1\}^n$ is a Cartesian product $R = R_1 \times \cdots R_n$, where each $R_i \subseteq \{-1, 0, 1\}$. We can compute the volume of rectangles easily: the volume of each $R_i$ is $\vol(R_i) = \sum_{r \in R_i} \alpha_{r}$, and volume of $R$ is given by $\vol(R) = \prod_{i=1}^n \vol(R_i)$. \Cref{eq:unbias-rewrite} gives an efficient method for computing the unbiased estimator of any efficiently decomposable function under $\mathcal{R}_n$. The next three sections will use this technique, with different choices of bases, to give efficient unbiased estimators for specific functions.

\subsection{Decision Tree Regression} 

It is easy to see that a decision tree regression function with $s$ internal nodes partitions the input space $\Z^n$ into $s+1$ rectangles. This means it can be written as $f(x) = \sum_{i=1}^{s+1} c_i \mathbf{1}[x \in R_i]$ for rectangles $R_1, R_2, \ldots, R_n \subseteq \Z^n$. This global property is enough to imply local decomposability under rectangles, and allow us to obtain an efficient estimator: 
\begin{theorem}\label{thm:ub-dt}
    A decision tree regression function $f$ with $s$ internal nodes is $s+1$-decomposable under $\mathcal{R}_n$, and its corresponding estimator is
    \[
        g(x) = \sum_{i=1}^{s+1} c_i \vol((R_i-x) \cap\{-1, 0, 1\}^n ).
    \]
\end{theorem}
Since the volume of a rectangle may be computed in $O(n)$ time, we can compute the unbiased estimator in $O(ns)$ time.
\begin{proof}
    For $\xi \in \{-1, 0, 1\}^n$, the decision tree regression is given by $f(x+\xi) = \sum_{i=1}^{s+1} c_i\mathbf{1}[x + \xi \in R_i] = \sum_{i=1}^{s+1} c_i\mathbf{1}[\xi \in R_i - x]$. As each $R_i-x$ is not necessarily a subset of $\{-1, 0, 1\}^n$, we express the function equivalently as $\sum_{i=1}^{s+1} c_i\mathbf{1}[\xi \in (R_i - x) \cap \{-1, 0, 1\}^n]$.
    The unbiased estimator then follows from \eqref{eq:unbias-rewrite}.
\end{proof}

\subsection{Order Statistics}
The minima and maxima of $x$ also admit efficient local decompositions under $\mathcal{R}_n$. This is because, for the points in $x + \{-1, 0, 1\}^n$, the max or min can take on three values, which can be written as a disjunction (sum) of three cases, each of which is a conjunction (rectangle). For example, in the case where $\min(x+\xi) = \min(x)+1$, it holds that $\xi_i = +1$ for all $i$ where $x_i = \min(x)$ and $\xi_i \in \{0, 1\}$ for all $i$ such that $x_i = \min(x)+1$. The other cases can be expressed similarly. This yields the following:
\begin{theorem}\label{thm:unbiased-ex}
    The minimum and maximum functions are $3$-locally decomposable under $\mathcal{R}_n$, and the corresponding estimator is $\min(x)-1 + (\alpha_0 + \alpha_1)^{\underline{n}_0} + (\alpha_1)^{\underline{n}_0}(\alpha_0 + \alpha_1)^{\underline{n}_1}$ for the minimum and $\max(x)+1 - (\alpha_0 + \alpha_{-1})^{\overline{n}_0} - (\alpha_{-1})^{\overline{n}_0}(\alpha_0 + \alpha_{-1})^{\overline{n}_1},$ for the maximum, where $\underline{n}_0$ (resp. $\overline{n}_0$) is the number of times in $x$ that $\min(x)$ (resp. $\max(x)$) appears, and $\underline{n}_1$ (resp. $\overline{n}_1$) is the number of times in $x$ that $\min(x)+1$ (resp. $\max(x)-1$) appears.
\end{theorem}
The above estimators can be computed in $O(n)$ time. 
\begin{proof}
    First, we will start with the minimum. Let $k = \min(x)$. Observe that $\min(x+\xi)$ is equal to either $k+1$, $k$ or $k-1$ when $\xi \in \{-1,0,1\}^n$. Let $A = \{i \in [n]: x_i = k\}$ and $B = \{i \in [n]: x_i = k+1\}$. We have $\min(x+\xi) = k-1$ if and only if at least one index of in $A$ satisfies $\xi_i = -1$. This set can be written as a difference of rectangles $R_1 \setminus R_2$, where $R_1 = \{-1,0,1\}^n$ and $R_2 = \{0,1\}^A \times \{-1,0,1\}^{[n] \setminus A}$.
    
    Next, $\min(x+\xi) = k$ if $\xi_i \geq 0$ for all indices in $A$, and $\xi_i = 0$ for at least one index in $A$, or $\xi_i = 1$ for all indices in $A$ and $\xi_i = -1$ for at least one index in $B$. The first set of indices can be written as $R_2 \setminus R_3$, where $R_3 = \{1\}^A \times \{-1, 0, 1\}^{[n] \setminus A}$. The second set of indices can be written as $R_3 \setminus R_4$, where $R_4 = \{1\}^A \times \{0,1\}^B \times \{-1,0,1\}^{[n] \setminus A \setminus B}$. Finally, $\min(x+\xi) = k+1$ precisely when $\xi \in R_4$. 
    This establishes that $\min$ can be written as 
    \begin{align*}
        \min(x + \xi) &= (k-1) \mathbf{1}[\xi \in R_1] - (k-1) \mathbf{1}[\xi \in R_2] + k \mathbf{1}[\xi \in R_2] - k \mathbf{1}[\xi \in R_3] \\
        &\qquad + k\mathbf{1} [\xi \in R_3] - k \mathbf{1}[\xi \in R_4] + (k+1) \mathbf{1} [\xi \in R_4] \\
        &= (k-1)\mathbf{1}[\xi \in R_1] + \mathbf{1}[\xi \in R_2] + \mathbf{1}[\xi \in R_4].
    \end{align*}
    Therefore, $\min$ is $3$-decomposable. Thus, by~\eqref{eq:unbias-rewrite}, the unbiased estimator is 
    \begin{align*}
        (k-1)\vol(R_1) + \vol(R_2) + \vol( R_4) = k-1 + (\alpha_0 + \alpha_1)^{|A|} + (\alpha_1)^{|A|}(\alpha_0 + \alpha_1)^{|B|}.
    \end{align*}
    
    We take a similar approach for the maximum. Let $k' = \max(x)$, $A' = \{i \in [n] : x_i = k'\}$, $B = \{i \in [n] : x_i = k'-1\}$, $R_2' = \{-1,0\}^{A'} \times \{-1, 0, 1\}^{[n]\setminus A'}$, $R_3' = \{-1\}^A \times \{-1, 0, 1\}^{[n]\setminus A}$, and $R_4' = \{-1\}^A \times \{-1, 0\}^B \times \{-1, 0, 1\}^{[n] \setminus A \setminus B}$. We can show that $\max$ is $3$-locally decomposable as 
    \begin{align*}
        \max(x + \xi) &= (k'+1) \mathbf{1}[\xi \in R_1] - (k'+1) \mathbf{1}[\xi \in R_2'] + k' \mathbf{1}[\xi \in R_2'] - k'\mathbf{1}[\xi \in R_3'] \\
        &\qquad + k' \mathbf{1}[\xi \in R_3'] - k' \mathbf{1}[\xi \in R_4'] + (k'-1) \mathbf{1}[\xi \in R_4'] \\
        &= (k'+1) \mathbf{1}[\xi \in R_1] - \mathbf{1}[\xi \in R_2'] - \mathbf{1}[\xi \in R_4'].
    \end{align*}
    The unbiased formula is given by
    \[
        k'+1 - (\alpha_0 + \alpha_{-1})^{|A'|} - (\alpha_{-1})^{|A'|}(\alpha_0 + \alpha_{-1})^{|B'|}. \qedhere
    \]
\end{proof}

General order statistics, which return element of rank $i$ in $x$ for a given $i$, can also be computed efficiently. To do so, we require a more complex basis. For a given set $A \subseteq [n]$ and an integer $1 \leq a \leq |A|$, define the set
    \begin{gather*}
        \textstyle{S_{A,a}^{-} = \{\xi \in \{-1, 0, 1\}^A : \sum_{i=1}^n \mathbf{1}[\xi_i = -1] = a\}} \\
        \textstyle{S_{A,a}^{+} = \{\xi \in \{-1, 0, 1\}^A : \sum_{i=1}^n \mathbf{1}[\xi_i = +1] = a\}},
    \end{gather*}
    or the vectors where exactly $a$ of $n$ coordinates are equal to $1$. The volume of these sets are given by
    \begin{gather*}
        \vol(S_{A,a}^-) = \binom{|A|}{a} \alpha_{-1}^a(\alpha_0 + \alpha_1)^{|A|-a} \\
        \vol(S_{A,a}^+) = \binom{|A|}{a} \alpha_{+1}^a (\alpha_{-1} + \alpha_0)^{|A|-a}.
    \end{gather*}
    
    Our basis will be the following collection of sets
    \[
        \mathcal{S} = \left\{ S_{A,a}^+ \times S_{B,b}^- \times \{-1, 0, 1\}^{[n]\setminus A\setminus B} : A, B \subseteq [n], A \cap B = \emptyset, 0 \leq a \leq |A|, 0 \leq b  \leq |B| \right\} .
    \]
    These are sets with both rectangular and symmetric structure, and their volume can be computed by multiplying the above volume formulae. 

    We now derive an efficient unbiased estimator for the $i$th order statistic using $\mathcal{S}$ as the basis.
    
    \begin{theorem}\label{thm:unbiased-ord}
        The $i$th order statistic is $O(n^2)$-locally decomposable under $\mathcal{S}$, and the corresponding unbiased estimate is computable in $O(n^2)$ time.
    \end{theorem}

    \begin{proof}
    Define $k = x_{(i)}$, or the element in $x$ of rank $i$. Observe that $f$ will return either $k-1, k$, or $k+1$ in the set $x + \{-1, 0, 1\}^n$. Define the sets of indices $A = \{j \in [n] : x_j = k-1\}$, $B = \{j \in [n] : x_j = k\}$, and $C = \{j \in [n] : x_j = k+1\}$. Let $\ell_1, \ell_2$ denote the maximum index with rank $k-1, k$; i.e. $\ell_1 = \sum_{i=1}^n \mathbf{1}[x_i \leq k-1]$ and $\ell_2 = \sum_{i=1}^n \mathbf{1}[x_i \leq k]$. Define $d_1 = i - \ell_1$ and $d_2 = \ell_2-i+1$. These are the ``distances'' that the rank $i$ element is from being equal to $k-1$ or $k+1$.
    
    We will begin by decomposing the set where $f(x+\xi) = k-1$. For $\xi \in \{-1, 0, +1\}^n$, let $a^+$ denote the number of indices $j \in A$ such that $\xi_j = +1$. Similarly, let $b^-$ denote the number of indices $j \in B$ such that $\xi_j = -1$. Thus, the element of rank $i$ will be $k-1$ if and only if $b^- - a^+ \geq d_1$. This means the $\xi$ where $f(x+\xi) = k-1$ can be decomposed as the following disjoint union of sets
    \[
        \bigsqcup_{b = 0}^{|B|} \bigsqcup_{a=0}^{b-d_1} S_{A,a}^+ \times S^{-}_{B,b} \times \{-1, 0, 1\}^{[n] \setminus A \setminus B}.
    \]
    Denote the above region as $R_{a,b} = S_{A,a}^+ \times S^{-}_{B,b} \times \{-1, 0, 1\}^{[n] \setminus A \setminus B}$. The volumes of the regions are 
    \begin{align*}
    \vol(R_{a,b}) &= 
    \vol(S^+_{A,a}) \vol(S_{B,b}^-) \\ &= \binom{|A|}{a} \binom{|B|}{b}\alpha_{1}^{a}(\alpha_{-1} + \alpha_0)^{|A|-a} \alpha_{-1}^{b} (\alpha_{0} + \alpha_1)^{|B|-b}.
    \end{align*}
    Similarly, let $b^+$ denote the number of indices $j \in B$ such that $x_{j} = 1$, and $c^-$ denote the number of indices $j \in C$ such that $x_j = -1$. 
    This means $f(x + \xi) = k+1$ if and only if $b^+ - c^-  \geq d_2$, corresponding to the region
    \[
        \bigsqcup_{b=0}^{|B|} \bigsqcup_{c = 0}^{b - d_2} S_{B, b}^+ \times S_{C, c}^- \times \{-1, 0, 1\}^{[n] \setminus B \setminus C}.
    \]
    Denote the above regions as $R'_{b,c}$. Their volumes are
    \begin{align*}
        \vol(R_{b,c}') &= \vol(S_{B,b}^+) \vol(S_{C,c}^-) \\ &= \binom{|B|}{b} \binom{|C|}{c} \alpha_{1}^b(\alpha_{-1} + \alpha_0)^{|B|-b}\alpha_{-1}^c (\alpha_0 + \alpha_1)^{|C|-c}.
    \end{align*}
    Finally, let $R = \{-1,0,+1\}^n$. We can represent $f(x + \xi)$ on the whole domain by starting with $k \mathbf{1}[\xi \in R]$, and then using the two unions above to either add $+1$ or $-1$. Formally, the representation is
    \[
        f(x + \xi) = k\mathbf{1}[\xi \in R] - \sum_{b=0}^{|B|} \sum_{a = 0}^{b - d_1} \mathbf{1}[ \xi \in R_{a,b}] + \sum_{b=0}^{|B|} \sum_{c=0}^{b-d_2} \mathbf{1}[\xi \in R_{c,b}'],
    \]
    showing that $f$ is $O(n^2)$-decomposable among symmetric rectangles. Using the volume formulae, the corresponding unbiased estimator is 
    \begin{multline*}
    g(x) = k - \sum_{b=0}^{|B|} \sum_{a = 0}^{b-d_1} \binom{|A|}{a} \binom{|B|}{b}\alpha_{1}^{a}(\alpha_{-1} + \alpha_0)^{|A|-a} \alpha_{-1}^{b} (\alpha_{0} + \alpha_1)^{|B|-b} \\ + \sum_{b=0}^{|B|} \sum_{c = 0}^{b-d_2} \binom{|B|}{b} \binom{|C|}{c} \alpha_{1}^b(\alpha_{-1} + \alpha_0)^{|B|-b}\alpha_{-1}^c (\alpha_0 + \alpha_1)^{|C|-c}.
    \end{multline*}
    To prove the running time guarantee, we can precompute the binomial coefficients up to $n$ in $O(n^2)$ time, as well as the exponents $\alpha_1^i$, $(\alpha_{-1}+\alpha_0)^i$ up to $n$. Then, the above estimator is a sum of $O(n^2)$ terms, resulting in a total time of $O(n^2)$.
\end{proof}

\subsection{Entropy}
The entropy takes the form $H(x_1, \ldots, x_n) = \sum_{i=1}^n h(\frac{x_i}{\sum_j x_j})$, where $h(z) = -z \log_2(z)$. As our results apply to the entire domain $\Z^n$, we need to extend the domain of the entropy function. To do so, we use an extension of the form
\[
    h(x, y) = \begin{cases} 0 & x \leq 0 \text{ or } y \leq 0 \\ 
    h(\min\{\frac{x}{y}, 1\}) & \text{otherwise}
    \end{cases}
\]
and write $H(x_1, \ldots, x_n) = \sum_{i=1}^n h(x_i, \sum_j x_j)$. There are many possible extensions of entropy to negative counts, and as noted in~\cite{calmon2025debiasing}, different extensions can produce different MSEs. We use the above extension because it is simple and bounded, which is sufficient for our purposes.

The structure of this function is particularly appealing for \Cref{def:decomp} because, on the domain $x + \{-1, 0, 1\}^n$, each $h(x_i, \sum_j x_j)$ depends only on $x_i + \xi_i$ and $\sum_j x_j + \xi_j$. This means that for each $i$, $h(x_i, \sum_j x_j)$ is constant on the set $P^i_{a,b} = \{\xi \in \{-1, 0, 1\}^n : \xi_i = a, \sum_j \xi_j = b\}$. This suggests the following basis:
\[
    \mathcal{P}_n = \{P_{a,b}^i : 1 \leq i \leq n, a \in \{-1, 0, 1\}, b \in \{-n, -n+1, \ldots, n\} \}.
\]
Furthermore, each $P_{a,b}^i$ is a highly symmetric set, and its volume can be computed efficiently using trinomial coefficients as follows:

\begin{lemma}\label{lem:vol-entropy}
    The volume of each $P_{a,b}^i$ satisfies
    \[
        \vol(P_{a,b}^i) = \sum_{\substack{0 \leq r \leq n \\ b+r \text{ is even}} } \alpha_{0}^{n-r} \alpha_1^{r}\binom{n-1}{n-1-r+|a|,\frac{b-a+r-|a|}2,\frac{r-|a|-b+a)}2}.
    \]
\end{lemma}
In the above formula, $\binom{n}{a,b,c}$ with integers $a+b+c = n$ denotes the trinomial coefficient given by $\frac{n!}{a!b!c!}$. If any of $a,b,c$ are negative, we extend the notation to evaluate to $0$.
\begin{proof}
        The product $\prod_{j=1}^n \alpha_{\xi_j}$ simplifies to $\alpha_{-1}^{n_{-1}}\alpha_0^{n_0} \alpha_1^{n_1}$, where $n_{-1}, n_0, n_1$ is the number of $-1$, $0$s, and $1$s in $\xi$, respectively. 
    Since $\alpha_{-1} = \alpha_1$, this simplifies to $\alpha_0^{n_0} \alpha_1^{n-n_0}$, or equivalently $\alpha_0^{n-\|\xi\|_1} \alpha_1^{\|\xi\|_1}$. 
    Let $l_1(r) = \{\xi \in \{-1, 0, 1\}^n : \|\xi\|_1 = r\}$ denote the $l_1$ shell of radius $r$. 
    We have
    \begin{align*}
        \vol(P_{a,b}^i) 
            &= \sum_{\xi \in P_{a,b}^i} \alpha_0^{n-\|\xi\|_1} \alpha_1^{\|\xi\|_1} \\
            &= \sum_{r= 0}^n \sum_{\xi \in P_{a,b}^i \cap l_1(r)} \alpha_0^{n-r} \alpha_1^{r} \\
            &= \sum_{r = 0}^n \alpha_0^{n-r} \alpha_1^{r} \lvert P_{a,b}^i \cap l_1(r) \rvert.
    \end{align*}
    Each vector $\xi \in P_{a,b}^i \cap l_1(r)$ satisfies $\xi_i = a$, $\sum_{j\neq i} \xi_j = b-a$, and $\sum_{j\neq i} |\xi_j| = r - |a|$. 
    An equivalent property for membership in $P_{a,b}^i \cap l_1(r)$ is that $\xi_i = a$ and for coordinates $i\neq j$, there are $n-1-r+|a|$ equal to $0$, $\frac{b-a+r-|a|}{2}$ equal to $+1$, and $\frac{-b+a+r-|a|}{2}$ equal to $-1$. Note that if $r,b$ have different parity, then $P_{a,b}^i \cap l_1(r) = \emptyset$. Thus, 
    \[
    |P_{a,b}^i \cap l_1(r)| = \begin{cases} \binom{n-1}{n-1-r+|a|, \frac{b-a+r-|a|}{2}, \frac{-b+a+r-|a|}{2}} & b+r \text{ is even} \\ 0 & \text{otherwise} \end{cases},\]
    and the result is shown.
\end{proof}

Having computed the volumes of each set in $\mathcal{P}_n$, we are able to give an efficient unbiased estimator for $H(x_1, \ldots, x_n)$.

\begin{theorem}\label{thm:entropy}
    The entropy function is $O(n^2)$-locally decomposable under $\mathcal{P}_n$, and its corresponding unbiased estimator is given by
    \[
        g(x_1, \ldots, x_n) = \sum_{i=1}^n \sum_{a \in \{-1, 0, 1\}} \sum_{b = -n}^n \vol(P^i_{a,b}) h\left(x_i + a, b+\sum_{j=1}^n x_j\right).
    \]
\end{theorem}
Using memoization, the $O(n^2)$ trinomial coefficients ahead of time requires $O(n^2)$ time. Having access to each coefficient, each $\vol(P_{a,b}^i)$ can be computed in $O(n)$ time. Since $\vol(P_{a,b}^i)$ does not depend on $i$, the volumes for each possible $a,b$ can all be computed in $O(n^2)$ time. The estimator in
Theorem~\ref{thm:entropy} is a sum over $O(n^2)$ terms, and thus the total running time is $O(n^2)$.

\begin{proof}
For all $\xi \in P_{a,b}^i$, $h(x+\xi_i, \sum_j x_j + \xi_j)$ takes the value $h(x_i + a, b + \sum_j x_j)$. Furthermore, the entire space $\{-1, 0, 1\}^n$ is partitioned by $P_{a,b}^i$ for $a \in \{-1, 0, 1\}$ and $b \in \{-n, \ldots, n\}$. Thus, we can write
\[
    h\left(x_i + \xi_i, \sum_{j=1}^n x_j + \xi_j\right) = \sum_{a \in \{-1, 0, 1\}} \sum_{b = -n}^n \mathbf{1}[\xi \in P_{a,b}^i]h\left(x_i + a, b + \sum_{j=1}^n x_j \right),
\]
and this gives the decomposition
\[
    H(x_1+ \xi_1, \ldots, x_n + \xi_n) = \sum_{i=1}^n\sum_{a \in \{-1, 0, 1\}} \sum_{b = -n}^n \mathbf{1}[\xi \in P_{a,b}^i] h\left(x_i + a, b + \sum_{j=1}^n x_j\right).
\]
The unbiased estimator follows by applying~\eqref{eq:unbias-rewrite}.
\end{proof}

\subsection{KL Divergence}\label{app:kl}
Assuming a domain of size $n$, one can compute the KL divergence between the empirical distributions of two different counts $x_1, \ldots, x_n$ and $y_1, \ldots, y_n$. This has the following formula: $KL(x_1, \ldots, x_n \| y_1, \ldots, y_n) = \sum_{i=1}^n kl(\frac{x_i}{\sum_j{x_j}} , \frac{y_i}{\sum_j y_j})$, where $kl(x,y) = x \log(\frac{x}{y})$. Again, because our results apply to arbitrary $x,y \in \Z^d$, we must extend $kl(x,y)$ to the entire integral domain. We do so by defining
\[
    kl(x_1, x_2, y_1, y_2) = \begin{cases} 0 & x_1\leq 0, x_2\leq 0, y_1\leq 0, \text{ or } y_2 \leq 0 \\
    kl(\min\{1, \frac{x_1}{x_2}\}, \min\{1, \frac{y_1}{y_2}\})
    \end{cases}
\]
which satisfies $KL(x_1, \ldots, x_n \| y_1, \ldots, y_n) = \sum_{i=1}^n kl(x_i, \sum_j{x_j} , y_i, \sum_j y_j)$.
Each term $kl(x_i, \sum_j{x_j} , y_i, \sum_j y_j)$ is actually constant on the set $P^i_{a_1,b_1} \times P^i_{a_2,b_2}$ for any $a_1,b_1,a_2,b_2$. Thus, we will use the following basis:
\[
    \mathcal{P}_n^2 = \{P_{a_1,b_1}^i \times P^i_{a_2,b_2} : 1 \leq i \leq n; a_1,a_2 \in \{-1, 0, 1\}; b_1,b_2 \in \{-n, \ldots, n\}\}.
\]
Furthermore, because $\vol(A \times B) = \vol(A) \vol(B)$, we have that $\vol(P_{a_1,b_1}^i \times P^i_{a_2,b_2}) = \vol(P_{a_1,b_1}^i) \vol(P^i_{a_2,b_2})$ and thus can be computed efficiently by Lemma~\ref{lem:vol-entropy}.
Using similar techniques as the entropy, we can obtain an unbiased estimator:
\begin{theorem}
    The KL divergence is $O(n^3)$-locally decomposable under $\mathcal{P}_n^2$, and its corresponding unbiased formula is     
    \begin{multline*}
        g(x_1, \ldots, x_n, y_1, \ldots, y_n) \\ = \sum_{i=1}^n \sum_{a_1,a_2 \in \{-1, 0, 1\}} \sum_{b_1 = -n}^n \sum_{b_2 = -n}^n \vol(P_{a_1,b_1}^i) \vol(P_{a_2,b_2}^i)kl\left(x_i+a_1, b_1+\sum_{j=1}^n{x_j}, y_i+a_2, b_2+\sum_{j=1}^n y_j \right).
    \end{multline*}
\end{theorem}

By memoizing $\vol(P^i_{a,b})$ for each possible $a,b$ (taking $O(n^2)$ time) and then computing the final sum over $O(n^3)$ terms, the above estimator has running time $O(n^3)$.

\begin{proof}
    For convenience, write $\xi_i' = \xi_{i+n}$. 
    For all $\xi \in P^i_{a_1,b_1} \times P^i_{a_2,b_2}$, the value of $kl(x_i+\xi_i, \sum_j{x_j} + \xi_j, y_i + \xi_i', \sum_j y_j + \xi_j')$ is $kl(x_i+a_1, b_1+\sum_j{x_j}, y_i+a_2, b_2+\sum_j y_j)$. The entire space $\{-1, 0, 1\}^{2n}$ is partitioned by the $P^i_{a_1,b_1} \times P^i_{a_2,b_2}$ for $a_1,a_2 \in \{-1, 0, 1\}$ and $b_1, b_2 \in \{-n, \ldots, n\}$. Thus, we can write
    \begin{multline*}
        kl\left(x_i + \xi_i, \sum_j x_j + \xi_j, y_i + \xi_i', \sum_j x_j'\right) \\ = \sum_{a_1,a_2 \in \{-1, 0, 1\}} \sum_{b_1 = -n}^n \sum_{b_2 = -n}^n \mathbf{1}[\xi \in P_{a_1,b_1}^i \times P_{a_2,b_2}^i]kl\left(x_i+a_1, b_1+\sum_j{x_j} , y_i+a_2, b_2+\sum_j y_j\right),
    \end{multline*}
    
    and the the KL divergence can be written as
    \begin{multline*}
        KL(x_1+\xi_1, \ldots, x_n+\xi_n \| y_1+\xi_1', \ldots, y_n+\xi_n') \\ = \sum_{i=1}^n \sum_{a_1,a_2 \in \{-1, 0, 1\}} \sum_{b_1 = -n}^n \sum_{b_2 = -n}^n \mathbf{1}[\xi \in P_{a_1,b_1}^i \times P_{a_2,b_2}^i]kl\left(x_i+a_1, b_1+\sum_j{x_j}, y_i+a_2, b_2+\sum_j y_j\right),
    \end{multline*}
    yielding the result.
\end{proof}

\subsection{Proof of \texorpdfstring{\Cref{thm:unbiased-prod}}{Theorem 14}}
\begin{proof}
    Rearranging the definition of $g$, 
    \begin{align*}
        g(y_1, \ldots, y_n) &= \sum\limits_{(\xi_1, \ldots, \xi_n) \in \{-1, 0, 1\}^n} f(y_1 + \xi_1, \ldots, y_n + \xi_n) \prod\limits_{j=1}^n \alpha_{\xi_j} \\
        &= \sum\limits_{(\xi_1, \ldots, \xi_n) \in \{-1, 0, 1\}^n} f_1(Y_1 + \xi^{(1)}) f_2(Y_2 + \xi^{(2)}) \prod_{\xi \in \xi^{(1)}} \alpha_{\xi} \prod_{\xi \in \xi^{(2)}} \alpha_\xi,
    \end{align*}
    where $\xi^{(1)}, \xi^{(2)}$ is the partitioning of $(\xi_1, \ldots, \xi_n)$ into the coordinates corresponding to $Y_1, Y_2$. We can also partition the sum this way---observe that the sum over $(\xi_1, \ldots, \xi_n) \in \{-1, 0, 1\}^n$ is the same as a sum of $(\vec{\xi}_1, \vec{\xi}_2) \in \{-1, 0, 1\}^{Y_1} \times \{-1, 0, 1\}^{Y_2}$. Thus, the above sum becomes 
    \begin{align*}
        &\;\sum\limits_{\xi^{(1)} \in \{-1, 0, 1\}^{Y_1}, \xi^{(2)} \in \{-1, 0, 1\}^{Y_2}} f_1(Y_1 + \xi^{(1)}) f_2(Y_2 + \xi^{(2)}) \prod_{\xi \in \xi^{(1)}} \alpha_{\xi} \prod_{\xi \in \xi^{(2)}} \alpha_\xi \\
        &= \left( \sum\limits_{\xi^{(1)} \in \{-1, 0, 1\}^{Y_1}} f_1(Y_1 + \xi^{(1)})  \prod_{\xi \in \xi^{(1)}} \alpha_{\xi} \right)\left(\sum\limits_{ \vec{\xi}_2 \in \{-1, 0, 1\}^{Y_2}}f_2(Y_2 + \xi^{(2)})\prod_{\xi \in \xi^{(2)}} \alpha_\xi\right) \\
        &= g_1(Y_1) g_2(Y_2).
        \qedhere
    \end{align*}
\end{proof}

\section{Empirical Results}\label{appendix:empirical-results}

In this appendix we present details of our empirical validation of our estimator~\Cref{eq:unbiased-multivariate}. 
We study four different types of non-linear functions:
\begin{itemize}
    \item Polynomials, for which we emphasize the advantage of the discrete Laplace mechanism over the classical Laplace mechanism for high values of $\varepsilon$.
    \item Empirical entropy, where we see a significantly improved mean square error for low- and medium values of $\varepsilon$.
    \item The partition function, for which we see improved estimation near the convergence radius boundary.
    \item Profiles, for which we give the first empirical results showing that accurate estimation is feasible in practice.
\end{itemize}

\subsection{Estimation of Polynomials}\label{sec:estimation-polynomials}

For the first application, we consider publishing the number of $k$-stars in a graph. A $k$-star is a subgraph consisting of a central node connected to $k$ neighbors. As such, it can be expressed as $\sum_i \binom{d_i}{k}$ with $d_i$ the degree of node $i$ and $\binom{y}{k}$ a polynomial in $y$.
We address the problem in the edge local differential privacy setting, where two adjacency vectors are considered neighbors if they have a Hamming distance of one. This enables us to publish each degree with Laplace or discrete Laplace with parameter $1 / \varepsilon$. Note that this is the same setting as the one considered in \cite{hillebrand2023unbiased} and \cite{imola_locally_2021}.

We evaluate this task on the Facebook dataset from \cite{leskovec2012learning}. This graph, where each node represents a user and edges represent friendships, contains 4039 nodes and 88,234 edges.

We compare four algorithms:
\begin{itemize}
    \item Unbiased DLap: the number of $k$-stars is computed using our algorithm by post-processing the degrees published with discrete Laplace.
    \item Naive DLap: the number of $k$-stars is naively computed on the degrees published with discrete Laplace.
    \item Unbiased Lap: the number of $k$-stars is computed using the algorithm from \cite{hillebrand2023unbiased} by post-processing the degrees published with continuous Laplace.
    \item LocalLapKStars: the method described in \cite{imola_locally_2021} that does not rely on post-processing.
\end{itemize}

\begin{figure}[t]
    \centering
    \includegraphics[width=0.45\linewidth]{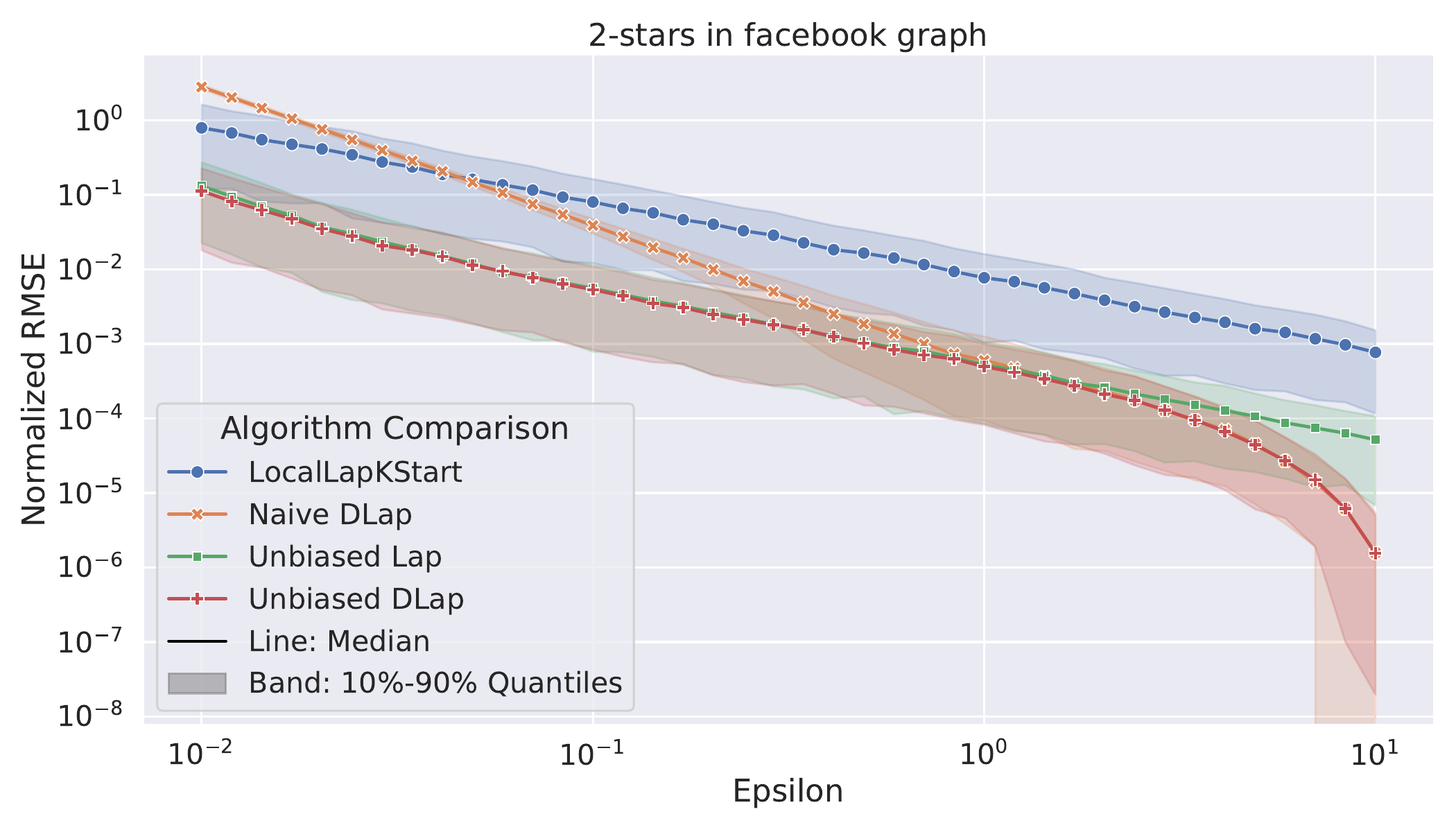}
    \includegraphics[width=0.45\linewidth]{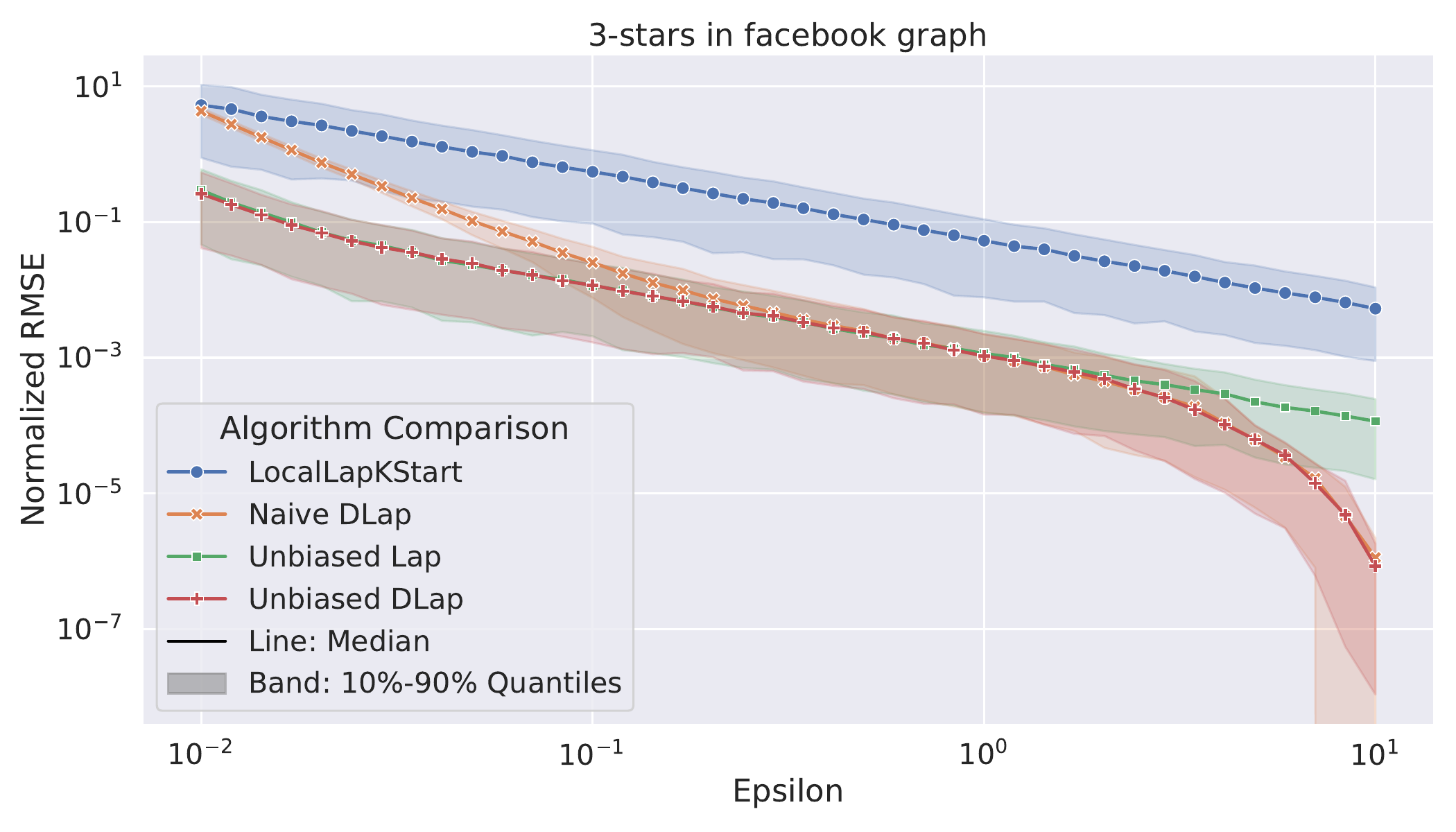}
    \caption{Root mean square error the number of $k$-stars (normalized) under varying privacy budgets.}
    \label{fig:kstars}
\end{figure}

The results, presented on Figure~\ref{fig:kstars}, show that our method outperforms the three other methods across all privacy budgets. We see that the gap with LocalLapKStars is large (at least two orders of magnitude) for all budgets. Naive DLap has a larger error with a sharper slope than Unbiased DLap until a threshold around $\varepsilon = 1$ after which the two methods have similar performances. When compared to Unbiased Lap, the MSE follows the one of Unbiased DLap until $\varepsilon = 1$ after which the error of Unbiased DLap shifts downward with an exponential decay.

These results show that our method obtains the best performance across all budgets with the biggest improvement over Unbiased Lap being in the low privacy regime, and the one over Naive DLap being in the high privacy regime. This shows that both the Discrete Laplace distribution and the unbiasing are necessary to obtain optimal performances.

\subsection{Estimation of Entropy}\label{sec:estimation-bounded-entropy}
In the second application, we study the estimation of the entropy of the distribution given by a histogram.
We focus on the \emph{bounded setting}, i.e., we assume that the sum of all columns $s$ is a publicly known value and that two histograms are neighbors if they have a $\ell_1$-sensitivity of 2.
The data consists of a histogram with $n \in \N^*$ columns, where the number of elements $x_i$ in column $i \in [n]$ gets published with discrete Laplace calibrated with sensitivity 2. 
The goal is to estimate, from the published counts, the value of the entropy
\[
    \sum_{i=1}^n \frac{x_i}{s} \log\frac{s}{x_i}.
\]

We evaluate both the naive and the unbiased estimator on the NIPS Bag of Words dataset \citep{newman2008bag}. 
To make the entropy function defined on $\mathbb{Z}^n$, as required for our estimator we replace the term $\frac{x_i}{s} \log\frac{s}{x_i}$ by zero when $x_i$ is not positive.
This dataset consists of a collection of 1,500 research papers available as a count histogram of words on a dictionary of 12,419 words. 
We excluded 9 documents that contained at most 5 words.
The results are presented in Figure~\ref{fig:bow-nips}. We can see that the unbiased method always outperforms the naive estimator, with the greatest improvement being for the low $\varepsilon$ regime.

\begin{figure}[t]
    \centering
    \includegraphics[width=0.6\linewidth]{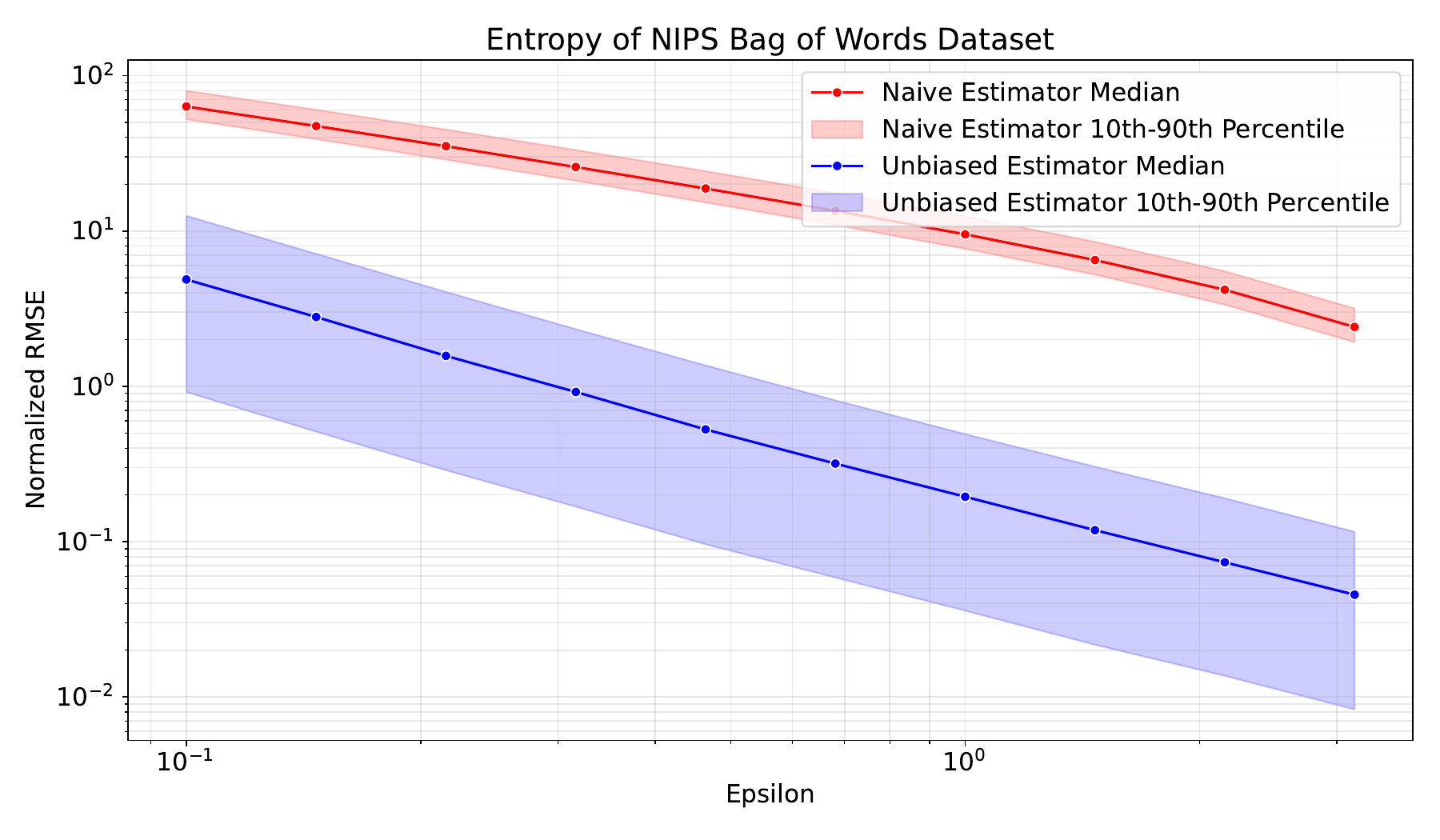}
    \caption{Mean square error of the entropy estimation on the bag of words of NIPS articles. Quantiles are shown over about 30,000 runs for each $\varepsilon$ (20 repetitions over 1491 histograms).}
    \label{fig:bow-nips}
\end{figure}

\subsection{Estimation of Partition Functions}\label{sec:partition-function}
We evaluate the estimators on the task of computing the partition function
$Z(t)=\sum_i \exp(t x_i)$ from a privatized histogram $\tilde{x}=x+\DLap(e^{-\varepsilon})$.
Partition functions are ubiquitous in machine learning: they normalize exponential-family models, energy-based models, and Gibbs distributions, and their logarithms encode quantities such as likelihoods, free energies, and cumulant-generating functions. 
Estimating $Z(t)$ from a private histogram therefore provides a simple test case for nonlinear post-processing of privatized data, while retaining the essential difficulty that exponentiation amplifies the bias introduced by additive noise.

To make estimation feasible we limit $t$ to an interval in which the naive estimator $\tilde{Z}(t) = \sum_i \exp(t \tilde{x}_i)$ has a well-defined expectation.
We compare $\tilde{Z}(t)$ to the unbiased estimator derived in Section~\ref{sec:DebiasDLap}, choosing parameter $\varepsilon$ small enough that the former has a visible bias.
To see how close the unbiased estimator is to what we might hope for, we also compare to a central model baseline that has access to the true histogram $x$ and needs to output \emph{a single value} of the partition function.

\paragraph{Datasets and results.}
Experiments are conducted on both synthetic and real datasets, varying the parameter $t$ and observing the distribution of estimates over $200$ noise realizations.
The datasets cover two synthetic frequency tables and two real histograms. 
The synthetic uniform table has approximately equal frequencies across the domain, while the synthetic Zipf table has a heavy-tailed frequency profile. 
The ExAC histogram records low allele-count genetic variants, restricted to non-TCGA variants with allele count at most $10$. 
The Shakespeare dataset is an expanded word-frequency profile derived from Shakespeare text, giving a real heavy-tailed distribution with a much larger dynamic range.

Figure~\ref{fig:uniform} (synthetic uniform data) shows that the naive estimator exhibits a clear positive bias that grows rapidly with $t$, leading to large relative error for moderate values of $t$.
In contrast, the unbiased estimator closely tracks the true partition function across the entire range, with substantially reduced bias.
The variance, visualized by the $5$--$95\%$ bands, is moderately larger for the unbiased estimator but remains well-controlled.
This demonstrates that for smooth distributions, debiasing can effectively remove systematic error without introducing prohibitive variance.

For heavy-tailed data (synthetic Zipf, Figure~\ref{fig:zipf}), the qualitative picture is similar but more pronounced.
The naive estimator significantly overestimates $Z(t)$ even for small $t$, while the unbiased estimator remains centered around the true value.
However, the spread of the unbiased estimator increases for larger $t$, reflecting the sensitivity of the exponential function to noise in high-magnitude coordinates.
This highlights the expected bias–variance trade-off: removing bias incurs increased variance, especially for skewed distributions.

Results on real datasets further confirm these trends.
On the ExAC histogram (Figure~\ref{fig:exac}), the unbiased estimator achieves near-zero relative error across all tested values of $t$, while the naive estimator accumulates a steadily increasing bias.
Similarly, on the Shakespeare dataset (Figure~\ref{fig:shakespeare}), which has large dynamic range, the unbiased estimator remains accurate even when $Z(t)$ spans many orders of magnitude, whereas the naive estimator again shows systematic overestimation.

Overall, the experiments validate the theoretical claims: the unbiased estimator successfully eliminates bias for nonlinear functions, and although it may incur higher variance, its mean accuracy is consistently superior across distributions.
In regimes where multiple independent estimates are aggregated, as for the partition function, this unbiasedness is particularly important, as it prevents the linear accumulation of error inherent to biased estimators.

\begin{figure}[t]
\centering
\begin{subfigure}{0.48\textwidth}
    \centering
    \includegraphics[width=\textwidth]{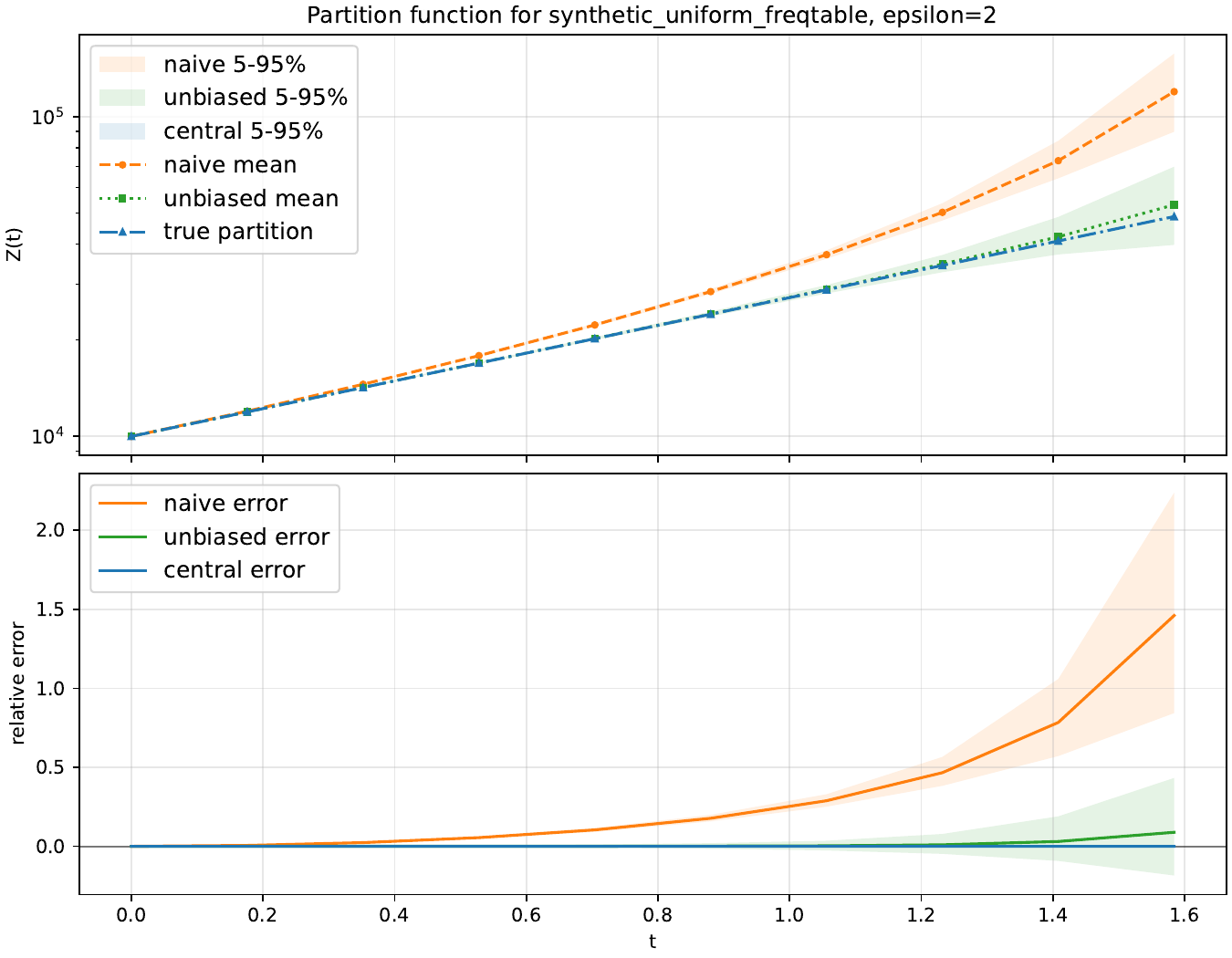}
    \caption{Synthetic uniform data.}
    \label{fig:uniform}
\end{subfigure}\hfill
\begin{subfigure}{0.48\textwidth}
    \centering
    \includegraphics[width=\textwidth]{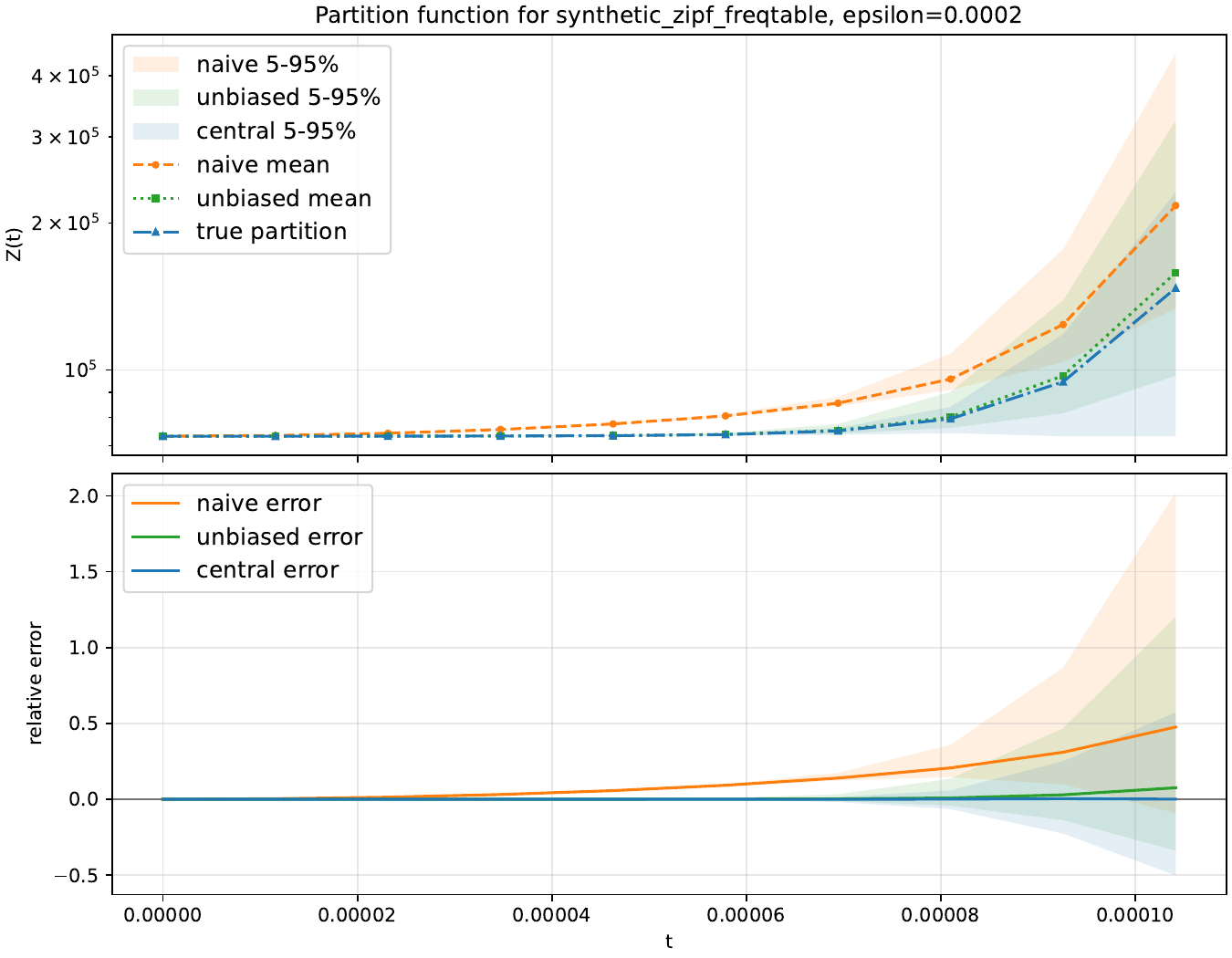}
    \caption{Synthetic Zipf data.}
    \label{fig:zipf}
\end{subfigure}
\caption{Partition function estimation on synthetic datasets.}
\label{fig:partition-synthetic}
\end{figure}

\begin{figure}[t]
\centering
\begin{subfigure}{0.48\textwidth}
    \centering
    \includegraphics[width=\textwidth]{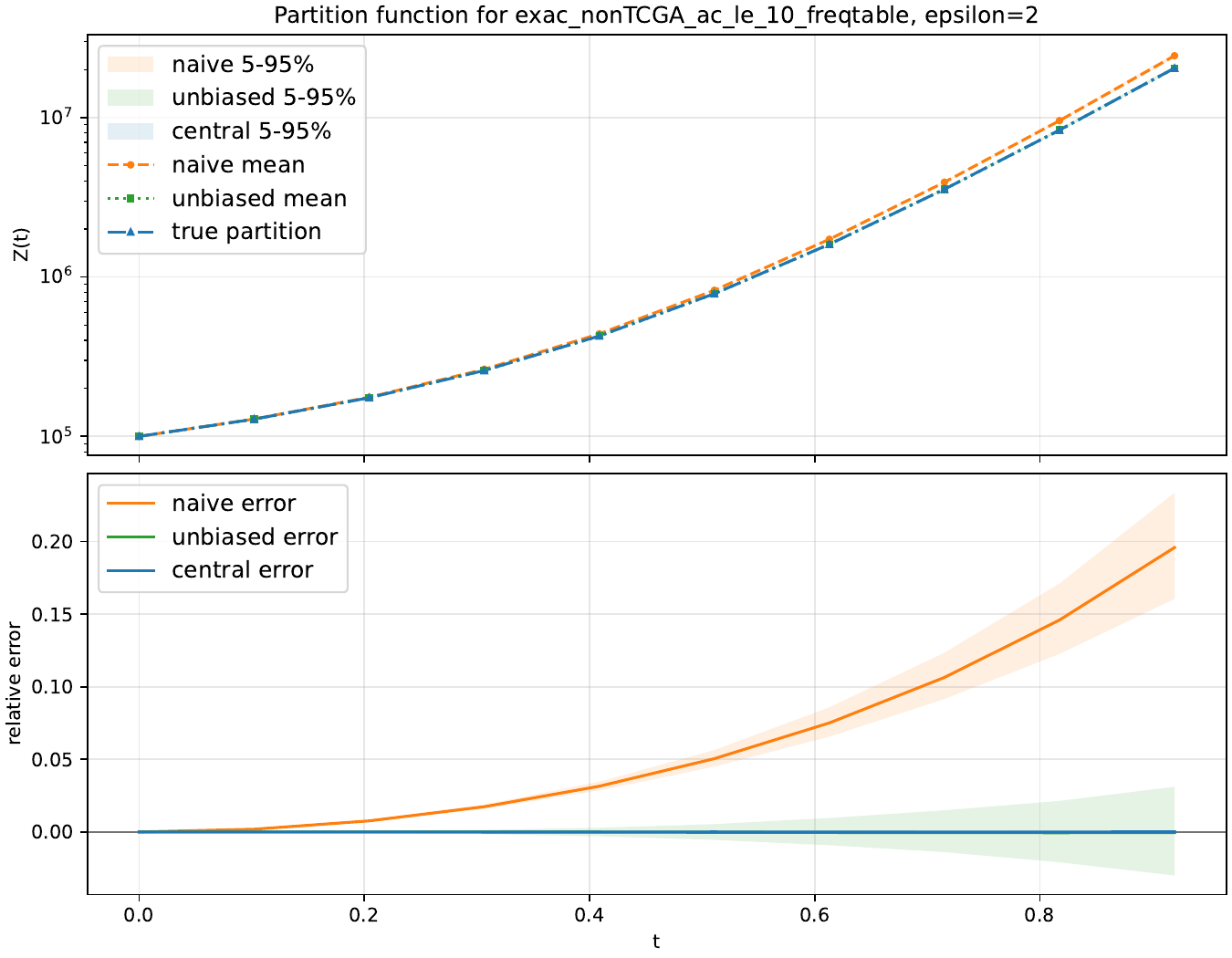}
    \caption{ExAC data.}
    \label{fig:exac}
\end{subfigure}\hfill
\begin{subfigure}{0.48\textwidth}
    \centering
    \includegraphics[width=\textwidth]{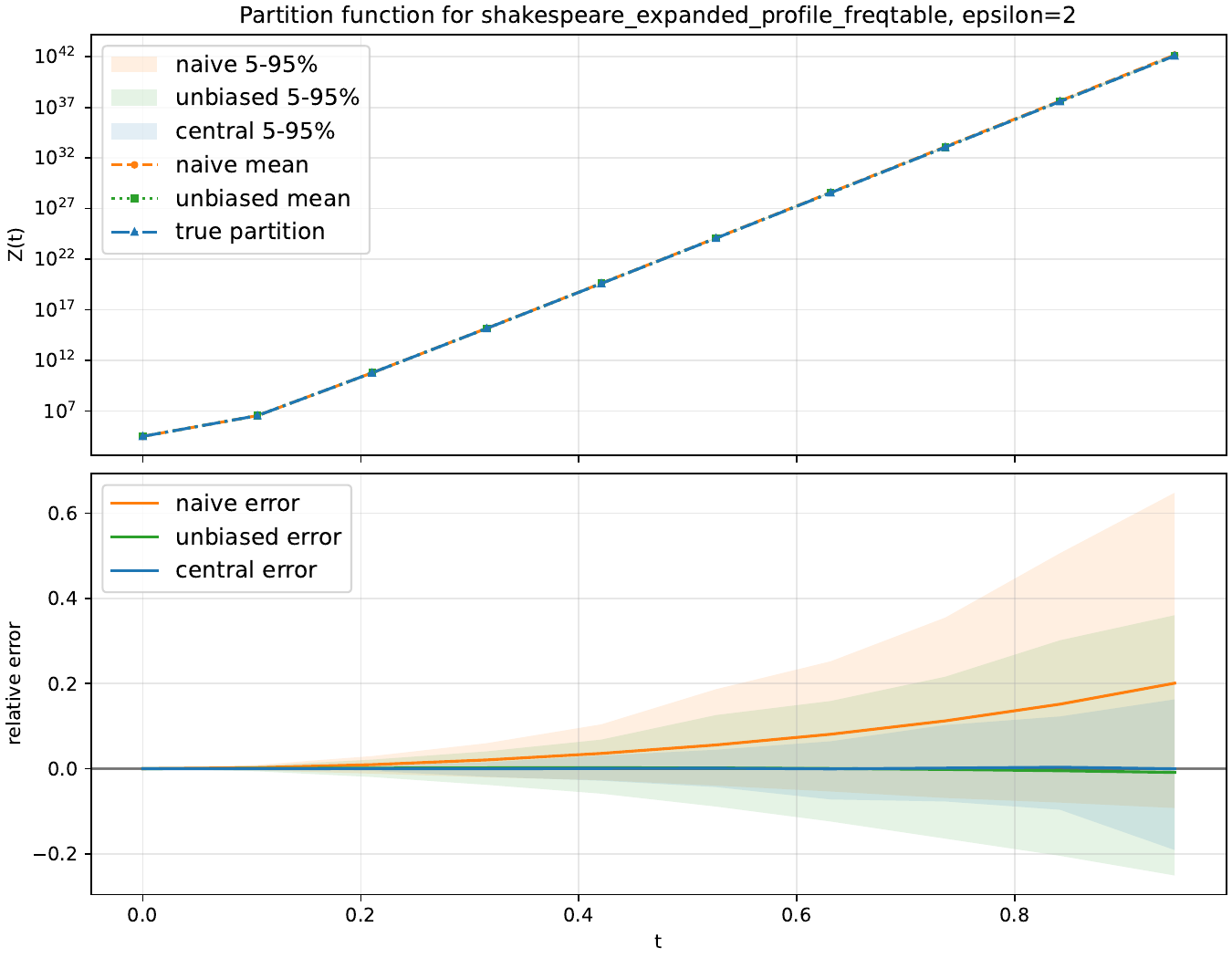}
    \caption{Shakespeare data.}
    \label{fig:shakespeare}
\end{subfigure}
\caption{Partition function estimation on real datasets.}
\label{fig:partition-real}
\end{figure}

\subsection{Profile Estimation}\label{sec:profile-estimation}
We next evaluate the estimators on profile estimation. Given a histogram $x$, the profile records, for each frequency $k$, the fraction of domain elements with count exactly $k$. Equivalently, for each $k$ we estimate
\[
    \phi_k(x) = \frac{1}{n}\sum_{i=1}^n \mathbf{1}[x_i=k].
\]
This is a natural discrete analogue of estimating the shape of a distribution from privatized counts, and it is also a basic primitive in distribution testing, species estimation, and frequency-statistics tasks. 
Since each coordinate contribution is an indicator function, profile estimation is a direct test of the local debiasing formula for nonsmooth functions.

Private profile estimation has been studied theoretically in several settings, including the sketch post-processing approach of~\citet{wu2024profile} and the shuffle model anonymized histogram estimators of~\citet*{ghazi2022anonymized}.
An unbiased estimator for this setting can easily be obtained from the threshold estimator of~\citep{ghazi2022anonymized}, so our estimator is not algorithmically novel in this setting, but we include empirical results for completeness.
(Prior works on private profile estimation did not, to our knowledge include empirical results.)

We compare the naive plug-in estimator, which computes the profile of the privatized histogram $\tilde{x}=x+\DLap(e^{-\varepsilon})$, to the unbiased estimator obtained by applying the univariate debiasing operator to each indicator $\mathbf{1}[x_i=k]$ and then averaging over coordinates. 
As in the partition-function experiments, we include a central model baseline that has access to the true histogram and outputs a single profile estimate. 
Since the algorithms in~\citep{wu2024profile,ghazi2022anonymized} are complex with unspecified implementation details we do not include them in the comparison.
All experiments use $\varepsilon=1$ and observe the distribution of estimates over $200$ noise realizations.

\paragraph{Datasets and results.} We use the same four datasets described in Section~\ref{sec:partition-function}. 
In the profile-estimation task, these datasets induce profiles with different shapes: the synthetic uniform table concentrates all mass at frequency~1, the synthetic Zipf and ExAC datasets have decreasing profiles over small frequencies, and the Shakespeare profile is more irregular and extends over a much wider range of frequencies.

Figure~\ref{fig:profile-synthetic} shows the results on synthetic data. On the uniform table, the naive estimator spreads mass away from the true frequency, severely underestimating the dominant profile entry and creating artificial mass at neighboring frequencies. The unbiased estimator removes this smoothing bias and closely follows the true profile. On the Zipf table, the naive estimator again underestimates the mass at frequency $1$ and overestimates several larger frequencies, whereas the unbiased estimator remains centered near the true profile.

Figure~\ref{fig:profile-real} shows the corresponding real-data experiments. 
On ExAC, the naive estimator has a large negative bias at frequency $1$ and a positive bias at nearby frequencies, while the unbiased estimator is nearly indistinguishable from the true and central profiles. 
On Shakespeare, both estimators track the broad shape of the profile, but the naive estimator shows oscillatory bias across frequencies, whereas the unbiased estimator remains centered with smaller mean error.

Overall, our empirical investigation of profile estimation demonstrate that our post-processing method can be effective for discontinuous statistics such as profiles.
The improvement is most visible at low frequencies, where discrete Laplace noise tends to move counts between adjacent bins and therefore introduces systematic smoothing in the naive profile.

\begin{figure}[t]
\centering
\begin{subfigure}{0.48\textwidth}
    \centering
    \includegraphics[width=\textwidth]{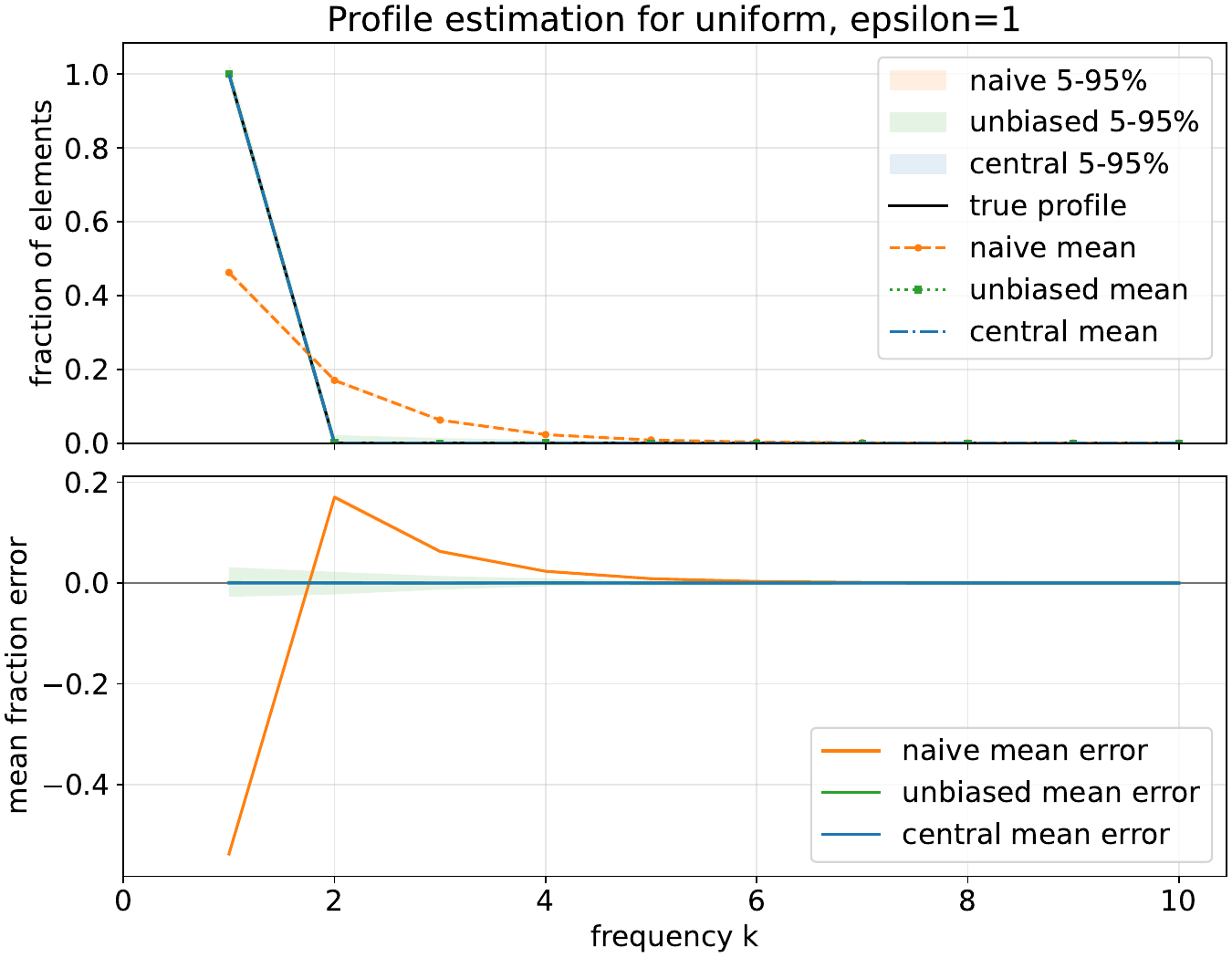}
    \caption{Synthetic uniform data.}
    \label{fig:profile-uniform}
\end{subfigure}\hfill
\begin{subfigure}{0.48\textwidth}
    \centering
    \includegraphics[width=\textwidth]{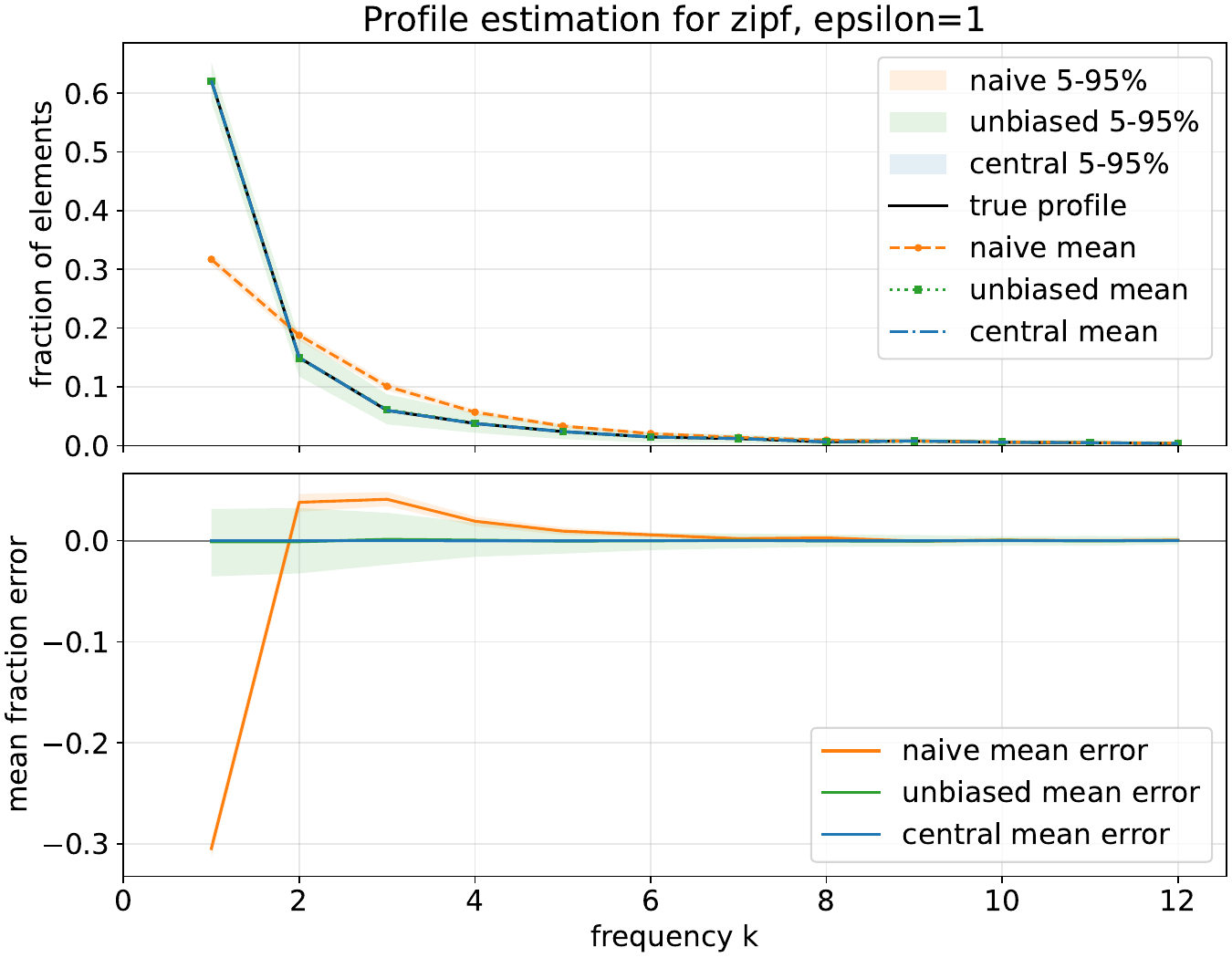}
    \caption{Synthetic Zipf data.}
    \label{fig:profile-zipf}
\end{subfigure}
\caption{Profile estimation on synthetic datasets.}
\label{fig:profile-synthetic}
\end{figure}

\begin{figure}[t]
\centering
\begin{subfigure}{0.48\textwidth}
    \centering
    \includegraphics[width=\textwidth]{figures/exac_nonTCGA_ac_le_10_freqtable_profile_sweep.pdf}
    \caption{ExAC data.}
    \label{fig:profile-exac}
\end{subfigure}\hfill
\begin{subfigure}{0.48\textwidth}
    \centering
    \includegraphics[width=\textwidth]{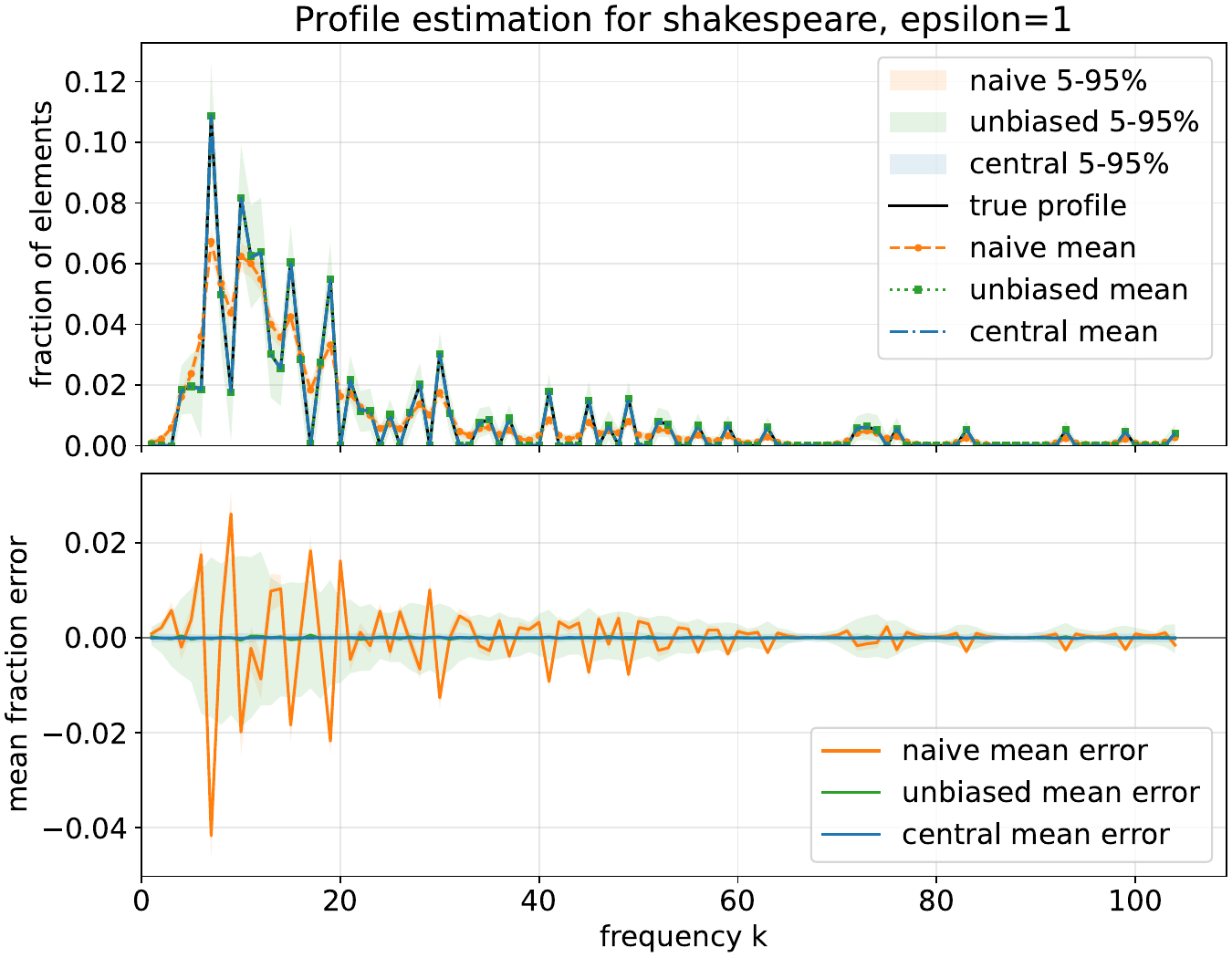}
    \caption{Shakespeare data.}
    \label{fig:profile-shakespeare}
\end{subfigure}
\caption{Profile estimation on real datasets.}
\label{fig:profile-real}
\end{figure}

\end{document}